\documentclass[11pt]{article}

\usepackage[T1]{fontenc}
\usepackage{amsmath,amsfonts,amssymb,amsthm}
\usepackage{xspace}
\usepackage{xcolor}
\usepackage{algorithmic}
\usepackage{bm}
\usepackage{booktabs}
\usepackage{fullpage}
\usepackage[top=0.88in,bottom=0.88in,left=0.88in,right=0.88in]{geometry}
\usepackage{nicefrac}
\usepackage{textcomp}
\usepackage{stmaryrd}
\usepackage{mathrsfs}
\usepackage{paralist}
\usepackage{enumitem}
\setlist[itemize]{leftmargin=*}
\setlist[enumerate]{leftmargin=*}
\usepackage{comment}
\usepackage{tikz}
\usepackage[font=footnotesize,labelsep=space,labelfont=bf]{caption}

\usepackage{url}

\usepackage[pdftex,bookmarks=true,pdfstartview=FitH,colorlinks,linkcolor=blue,filecolor=blue,citecolor=blue,urlcolor=blue,pagebackref=true]{hyperref}
\newtheorem{thm}{Theorem}
\newtheorem{defn}{Definition}
\newtheorem{lem}{Lemma}

\newcommand{\namedref}[2]{\hyperref[#2]{#1~\ref*{#2}}}

\newcommand{\Sectionref}[1]{\namedref{Section}{sec:#1}}
\newcommand{\Subsectionref}[1]{\namedref{Subsection}{subsec:#1}}
\newcommand{\Appendixref}[1]{\namedref{Appendix}{app:#1}}
\newcommand{\Theoremref}[1]{\namedref{Theorem}{thm:#1}}

\newcommand{\Lemmaref}[1]{\namedref{Lemma}{lem:#1}}

\newcommand{\Figureref}[1]{\namedref{Figure}{fig:#1}}
\newcommand{\Equationref}[1]{\namedref{Equation}{eq:#1}}
\newcommand{\Footnoteref}[1]{\namedref{Footnote}{foot:#1}}
\newcommand{\Pageref}[1]{\hyperref[#1]{page~\pageref*{#1}}}

\definecolor{darkred}{rgb}{0.5, 0, 0} 
\definecolor{darkblue}{rgb}{0,0,0.5} 
\hypersetup{
    colorlinks=true,
    linkcolor=darkred,
    citecolor=darkblue,
    urlcolor=darkblue   
}

\renewcommand{\ij}{\ensuremath{\vec{ij}}\xspace}
\newcommand{\ji}{\ensuremath{\vec{ji}}\xspace}

\newcommand{\RI}{\ensuremath{RI}\xspace}
\newcommand{\out}[1]{\ensuremath{\Pi^{\text{out}}_{#1}}\xspace}

\newcommand{\X}{\ensuremath{\mathcal{X}}\xspace}
\newcommand{\Y}{\ensuremath{\mathcal{Y}}\xspace}
\newcommand{\Z}{\ensuremath{\mathcal{Z}}\xspace}
\newcommand{\equi}{\ensuremath{\triangleleft}\xspace}

\renewcommand{\paragraph}[1]{\smallskip\noindent{\bf #1}~}

\begin{document}
\title{On the Communication Complexity of\\Secure Computation}
\author{Deepesh Data and Vinod M. Prabhakaran\\
School of Technology and Computer Science\\
Tata Institute of Fundamental Research\\
Mumbai 400 005
\and  
Manoj M. Prabhakaran\\
Department of Computer Science\\
University of Illinois, Urbana-Champaign\\
Urbana, IL 61801}

\maketitle

\begin{abstract}

Information theoretically secure multi-party computation (MPC) has been a central
primitive of modern cryptography. However, relatively little is known about
the communication complexity of this primitive.

In this work, we develop powerful information theoretic tools to prove lower
bounds on the communication complexity of MPC. We restrict ourselves to a
concrete setting involving 3-parties, in order to bring out the power of
these tools without introducing too many complications. Our techniques
include the use of a data processing inequality for {\em residual
information} --- i.e., the gap between mutual information and
G\'acs-K\"orner common information, a new {\em information inequality} for
3-party protocols, and the idea of {\em distribution switching} by which
lower bounds computed under certain worst-case scenarios can be shown to
apply for the general case.

Using these techniques we obtain tight bounds on communication complexity by MPC protocols for various interesting functions. In particular, we show concrete functions which have
``communication-ideal'' protocols, which achieve the minimum
communication simultaneously on all links in the network. Also, we obtain
the first {\em explicit} example of a function that incurs a higher communication
cost than the input length, in the secure computation model of Feige, Kilian
and Naor \cite{FeigeKiNa94}, who had shown that such functions exist. We also show that our communication bounds imply tight lower bounds on the amount of randomness required by MPC protocols for many interesting functions.

We identify a {\em multi-secret sharing} primitive that is interesting
on its own right, but also has the property that lower bounds on its share
sizes serve as lower bounds for communication complexity of MPC protocols.
While often the resulting bounds are tight, we can use our results to give a
concrete example where there is a gap between the share sizes and the
communication complexity.

\end{abstract}

\thispagestyle{empty}
\newpage

\thispagestyle{empty}
\tableofcontents
\thispagestyle{empty}
\pagestyle{empty}

\newpage
\setcounter{page}{1}

\pagestyle{plain}

\section{Introduction}
Information theoretically secure multi-party computation has been a central
primitive of modern cryptography. The seminal results of Ben-Or, Goldwasser,
and Wigderson~\cite{BenorGoWi88} and Chaum, Cr\'epeau, and
Damg\r{a}rd~\cite{ChaumCrDa88} showed that information theoretically secure
function computation is possible between parties connected by pairwise,
private links as long as only a strict minority may collude in the
honest-but-curious model (and a strictly less than one-third minority may
collude in the malicious model). Since then, several protocols have improved
the efficiency of these protocols.

However, relatively less is known about {\em lower bounds} on the amount of
{\em communication} required by a secure multi-party computation protocol,
with a few notable exceptions
\cite{Kushilevitz89,FranklinYu92,ChorKu93,FeigeKiNa94}. In fact,
\cite{IshaiKu04} shows that establishing strong communication lower bounds
(even with restrictions on the number of rounds) would imply breakthrough
lower bound results for other well-studied problems like private-information
retrieval and locally decodable codes. Further, due to known upper bounds on
the communication needed in a secure multi-party computation protocol
\cite{DamgardIs06}, such lower bounds would imply non-trivial circuit
complexity lower bounds --- a notoriously hard problem in theoretical
computer science. The goal of this work is to develop tools to tackle the
difficult problem of lower bounds for communication in secure multi-party
computation, even if they do not immediately have direct implications to
circuit complexity or locally decodable codes.%
%
%

It is instructive to compare the problem of communication complexity lower
bounds for secure multi-party computation with that when there is no
security requirement involved. This latter problem has been extensively
studied --- over the last three and a half decades, starting with
\cite{Yao79} --- resulting in a rich collection of results and techniques.
Unfortunately, many of the techniques in the communication complexity
setting are not relevant in the setting of secure computation:%
\footnote{Of course, communication complexity lower bounds continue to hold
for secure computation as well, but these bounds as such are (apparently)
very loose (since there is a trivial upper bound for communication
complexity, which is at most the size of all inputs and outputs).}
for instance, for communication complexity Yao's minimax theorem allows one
to consider only deterministic protocols with public randomness, but in the
secure computation setting, one must allow private randomness, and hence it
is not sufficient to consider only deterministic protocols. This rules out
several powerful combinatorial approaches from the communication complexity
literature. But over the last decade or so (see for example, \cite{KerenidisLLRX12} and
references therein), a slew of information theoretic tools have been
developed, which in many cases subsume more complicated combinatorial
approaches.

Following this lead, the approach we take in this work is to develop novel
{\em information-theoretic tools to obtain lower bounds on the communication
complexity of secure computation}. Indeed, the tools we develop and use
have connections with similar tools developed in the context of
communication complexity and related problems. In particular, all these
tools are related to notions of ``common information'' introduced by
G\'{a}cs-K\"{o}rner \cite{GacsKorner} and Wyner \cite{Wyner75}.%
\footnote{In communication complexity and related problems, the lower bound
techniques relate to Wyner common information~\cite{PPunderprep14,BraunPo13},
whereas the tools in this work are more directly related to
G\'{a}cs-K\"{o}rner common information. Wyner common information and
G\'{a}cs-K\"{o}rner common information have been generalized to a measure of
correlation represented as the ``tension region'' in
\cite{PrabhakaranPr12}.}

In this work we restrict our study to a concrete setting that brings out the
power of these tools without introducing too many additional complications.
Our setting involves 3 parties (with security against corruption of any
single party) of which only two parties have inputs, $X$ and $Y$, and only
the third party receives an output $Z$ as a (possibly randomized) function
of the inputs. This class of functions is similar to that studied in
\cite{FeigeKiNa94}, but our protocol model is more general (since it allows
fully interactive communication), making it harder to establish lower
bounds.

\subsection{Results and Techniques}
We study the setting shown in \Figureref{setup}. 
We obtain lower bounds
on the expected number of bits that need to be exchanged between each pair
of parties when securely evaluating a (possibly randomized) function of two
inputs so that Alice and Bob have one input each, and Charlie receives the
output. In fact, our bounds are on the entropy of the transcript between
each pair,%
\footnote{The entropy bounds translate to bounds on the expected number of
bits communicated, when we require that the messages on the individual links
are encoded using (possibly adaptively chosen) prefix-free codes. See
\Appendixref{communication-lowerbound-through-entropy-lowerbound}.}
and hence hold even when the protocol is amortized over several
instances with independent inputs. Further, these bounds do not depend on
the input distribution (as long as the distribution has full support) and
hold even if the protocol is allowed to depend on the input distribution.

At a high-level, the ingredients in deriving of our lower-bounds are the following:
\begin{itemize}
\item Firstly, we observe that, since Alice and Bob do not obtain any outputs, they are both forced to reveal their inputs fully (upto equivalent inputs) to the rest of the system. This implies that the transcripts of a secure computation form the shares of the inputs and outputs according to an appropriately defined ``correlated multiple secret sharing
scheme'' (CMSS).%
\footnote{We remark that our notion of multiple secret sharing schemes is different from that of \cite{BlundoSaCrGaVa94}, which (implicitly) required that secrets with different access structures be
independent of each other. In our case, $Z$ is typically strongly correlated with $X,Y$, often via a deterministic function.}
Hence, lower bounds on the entropies of the shares in a CMSS imply lower bounds on the entropies of the messages in a secure computation protocol. One can immediately obtain a na\"ive lower bound on the entropies of the shares in a correlated multiple secret sharing scheme: specifically, if $X,Y,Z$ are the secrets, and $M_{23}$ denotes the part of shares that is not available to a party who should learn only $X$, then we can see that $H(M_{23}) \ge H(Y,Z|X)$.%
\footnote{\label{foot:addition}
We point out a simple example for which one can obtain a tight bound from this na\"ive bound for CMSS: addition (in any group) requires one group element to be communicated between every pair of players, even with amortization
over several independent instances. Previous lower bounds for secure evaluation of addition (in any group) \cite{FranklinYu92,ChorKu93}, while considering an arbitrary number of parties, either restricted themselves to
bounding the {\em number of messages} required, or relied on non-standard security requirements. (For the 3-party case, for semi-honest security, results of \cite{FranklinYu92,ChorKu93} only imply that all three links should be used. \cite{FranklinYu92} did give a lower bound on the number of bits communicated as well, but this was shown only under a non-standard security requirement called {\em unstoppability}.)
}

We strengthen the na\"ive lower bounds by relying on a ``data-processing inequality'' for {\em residual information} --- i.e., the gap between mutual-information and (G\'{a}cs-K\"{o}rner) common information --- which lets us relate the
residual information between the shares to the residual information between the secrets. This bound is given in \Theoremref{prelim_lbs}.

\item We can further improve the above lower bounds using a new tool, called {\em distribution switching}. The key idea is that the security requirement forces the distribution of the transcript on certain links to be independent of the inputs. Hence, we can optimize our bounds over all input distributions having full support. Further, this shows that even if the protocol is allowed to depend on the input distribution, our bounds (which depend only on the function being evaluated) hold for every input distribution that has full support over the input domain. The resulting bound is summarised in \Theoremref{improved_prelim_lbs}.

\item As it turns out, CMSS lower bounds are in general weak, because a CMSS can in fact be strictly more efficient than a secure computation protocol that the CMSS problem is derived from (see \Appendixref{CMSS_sampling}). To
go beyond the CMSS bounds, we need to exploit the fact that in a protocol, the transcripts have to be generated by the parties interactively, rather than be created by an omniscient ``dealer''. An important technical contribution of this work is to provide a new tool towards this, in the form of a {\em new information inequality for 3-party interactive protocols}~(\Lemmaref{infoineq}). We use this to derive a bound (\Theoremref{intermediate_lbs}) that serves as an intermediate result for us.

Using the idea of {\em distribution switching}, we can significantly improve the above lower bounds by optimizing them using appropriate distributions of inputs. In fact, we can take the different terms in our bounds and {\em optimize each of them separately using different distributions over the inputs}. The resulting bounds (\Theoremref{lowerbound} and \Theoremref{conditionallowerbounds}) are often stronger than what can be obtained by considering a single input distribution for the entire expression.

\end{itemize}

The resulting bounds are summarized in \Theoremref{prelim_lbs}, \Theoremref{improved_prelim_lbs}, \Theoremref{lowerbound} and
\Theoremref{conditionallowerbounds}. While we restrict our attention to a 3
party setting, to the best of our knowledge, these are the first {\em
generic} lower bounds which apply to any function. To illustrate their use,
we apply them to several interesting example functions.
In particular, we show the following:
\begin{itemize}
\item We analyze secure protocols for two functions --
{\sc group-add, controlled-erasure} and {\sc remote-ot}  --
and, applying our lower bounds, show that these protocols
achieve {\em optimal communication complexity 
simultaneously on each link}. We call such a protocol a {\em
communication-ideal} protocol. We leave it open to characterize which
functions have communication-ideal protocols.

\item We use our lower bounds to establish a separation between secret
sharing and secure computation: we show that there exists a function (in
fact, the {\sc and} function) which has a CMSS scheme with a share strictly
smaller than the number of bits in the transcript on the corresponding link
in any secure computation protocol for that function. While such a
separation is natural to expect, we note that proving it requires exploiting
the properties of an interactive protocol.

\item We show an {\em explicit} deterministic function $f:\{0,1\}^n \times
\{0,1\}^n \rightarrow \{0,1\}^{n-1}$  which has a
communication-ideal protocol in which Charlie's total communication cost is
(and must be at least) $3n-1$ bits. In contrast, \cite{FeigeKiNa94} showed
that {\em there exist} functions $f:\{0,1\}^n \times \{0,1\}^n \rightarrow
\{0,1\}$, for which Charlie must receive at least $3n-4$ bits, if the
protocol is required to be in their non-interactive model.
(Note that our bound is incomparable to that of \cite{FeigeKiNa94}, since we
require the output of our function to be longer; on the other hand, our
bound uses an explicit function, and continues to hold even if we allow
unrestricted interaction.)

\item Our lower bounds for communication complexity also yield lower bounds on the amount of randomness needed in secure computation protocols. We analyze secure protocols for {\sc group-add, controlled-erasure, remote-ot} and {\sc sum}, and prove that these protocols are {\em randomness-optimal}, i.e., they use the least amount of randomness.
\end{itemize}

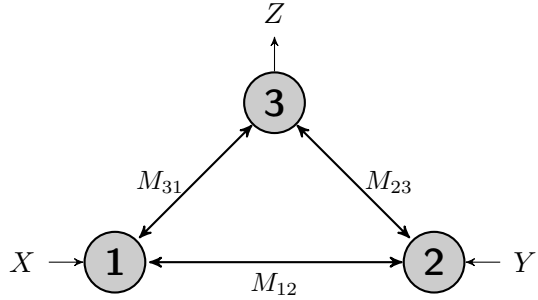
\begin{figure}[tb]

\centering
\usetikzlibrary{arrows}
\begin{tikzpicture}[<->,>=stealth',shorten >=1pt,auto,node distance=3cm,
  thick,main node/.style={circle,fill=black!20,draw,font=\sffamily\Large\bfseries}]

  \node[main node, pin={[pin edge={->,thin,black}]above:$Z$} ] (3) {3};
  \node[main node, pin={[pin edge={<-,thin,black}]left:$X$} ](1) [below left of=3] {1};
  \node[main node, pin={[pin edge={<-,thin,black}]right:$Y$} ] (2) [below right of=3] {2};
  
  \path[every node/.style={font=\sffamily\small}]
    (3) edge node [right] {$M_{23}$} (2)
        edge node [left] {$M_{31}$} (1)
    (2) edge node {$M_{12}$} (1);
    
\end{tikzpicture}
\caption{A three-party secure computation problem. Alice (party-1) has input
$X$ and Bob (party-2) has $Y$. We require that (i) Charlie (party-3) obtains
as output a randomized function of the other two parties' inputs,
distributed as $p_{Z|XY}$, (ii) Alice and Bob learn no additional information about
each other's inputs,
and (iii) Charlie learns nothing more about $X,Y$ than what is revealed by
$Z$. All parties can talk to each other, over multiple rounds over
bidirectional pairwise private links.} 
\label{fig:setup}
\end{figure}

\subsection{Related Work}
Communication complexity of multi-party computation without security
requirements has been widely studied since \cite{Yao79} (see
\cite{KushilevitzNi97book}), and more recently has seen the use of
information-theoretic tools as well, in \cite{ChakrabartiShWiYa01} and
subsequent works.  Independently, in the information theory literature
communication requirements of interactive function computation have been
studied (e.g.\ \cite{OrlitskyRoche}).

In secure multi-party computation, there has been a vast literature on
information-theoretic security, focusing on building efficient protocols, as
well as characterizing various aspects like corruption models that admit
secure protocols (e.g.\
\cite{BenorGoWi88,ChaumCrDa88,Chaum89,HirtMa97,FitziHiMa99,HirtLuMaRa12})
and the number of rounds of interaction needed (e.g.\
\cite{FischerLy82,GennaroIsKuRa01,FitziGaGoRaSr06,PatraChRaRa09,KatzKoKu09}),
In the computational security setting, \cite{NaorNi01} gave upper bounds on
the communication complexity of 2-party secure computation in terms of the
communication complexity without security requirements. In the
information-theoretic security setting, \cite{DamgardIs06} upper bounded the
communication complexity of multi-party secure computation in terms of the
circuit complexity of the computation.

But {\em lower-bounding communication complexity} has received much less
attention.  For 2-party secure computation with security against passive
corruption of one party (when the function admits such a protocol),
communication complexity was combinatorially characterized in
\cite{Kushilevitz89}. Franklin and Yung \cite{FranklinYu92} showed that the
{\em number of messages} used in a protocol must be quadratic in $n$, the
number of parties (if security against corrupting $t=\Omega(n)$ parties is
required).
Further \cite{FranklinYu92,ChorKu93} gave tight lower and upper bounds on
the number of messages needed for secure computation of the ``modulo-sum'' function
by $n$ parties; relying on a stronger corruption model (fail-stop
corruption), \cite{FranklinYu92} also argued a lower bound for the amortized
{\em communication complexity} of secure summation. \cite{FeigeKiNa94}
obtained a lower bound on the communication complexity for a restricted
class of 3-party protocols; along with positive results, they gave a modest
lower bound for communication needed for evaluating random functions in this
model. The difficulty of obtaining general lower bounds was pointed out in
\cite{IshaiKu04}, who related such lower bounds to lower bounds for locally
decodable codes and private information retrieval protocols. The connection
to private information retrieval protocols was recently used
in~\cite{BeimelIsKuKu14} to, among other things, derive the best known
general upper bound on communication for Boolean functions in the model
of~\cite{FeigeKiNa94}. Note that this upper bound is exponential in the number of
input bits compared to the lower bound of~\cite{FeigeKiNa94} which is only linear. The question of how much randomness is required for secure computation seems to have received even less attention; we are aware of \cite{KushilevitzMa97, BlundoSaPeVa99, GalRo05, AbbeLe14}.

Information-theoretic tools have been successfully used in deriving bounds
in various cryptographic problems like key agreement (e.g.\
\cite{MaurerWo03}), secure 2-party computation (e.g.\ \cite{DodisMi99}) and
secret-sharing and its variants (e.g.\ \cite{BeimelOr11} and
\cite{BlundoSaCrGaVa94}).  In this work, we rely on information-theoretic
tools developed in \cite{WolfWu08,PrabhakaranPr12}, which also considered
cryptographic problems. Some preliminary observations leading to this work
appeared in \cite{DPAllerton13} (as referenced at the appropriate points,
below).

\section{Preliminaries}

\paragraph{Notation.}
We write $p_X$ to denote the distribution of a discrete random variable
$X$; $p_X(x)$ denotes $\Pr[X=x]$. When clear from the context, the
subscript of $p_X$ will be omitted.
The conditional distribution denoted by $p_{Z|U}$ specifies
$\Pr[Z=z|U=u]$,
for each value $z$ that $Z$ can take and each
value $u$ that $U$ can take.
A {\em randomized function} of two variables,  is specified by
a probability distribution $p_{Z|XY}$, where $X,Y$ denote the two input
variables, and $Z$ denotes the output variable.
For a sequence of random variables $X_1,X_2,\ldots,$ we denote by $X^n$ the vector $(X_1,\ldots,X_n)$.

For random variables $T,U,V,$ we write the {\em Markov chain} $T-U-V$ to
indicate that $T$ and $V$ are conditionally independent conditioned on $U$: i.e.,
$I(T;V|U)=0$.  All logarithms are to the base 2. The binary entropy function
is denoted by $H_2(p)=-p\log p -(1-p)\log(1-p),\;p\in(0,1).$

\paragraph{Protocols.}
A $3$-party protocol $\Pi$ is specified by a collection of ``next message
functions'' $(\Pi_1,\Pi_2,\Pi_3)$ which probabilistically map a {\em state}
of the protocol to the next state (in a restricted manner), and output functions
$(\out1,\out2,\out3)$ used to define the outputs of the parties as
probabilistic functions of their views.  We shall also allow the protocol to
depend on the distribution of the inputs to the parties. (This would allow
one to tune a protocol to be efficient for a suitable input distribution.
Allowing this makes our lower bounds stronger; on the other hand, none of the
protocols we give for our examples require this flexibility. See discussion
in \Appendixref{discuss}.)

In \Appendixref{prelims} we formalize a well-formed 3-party
protocol.  Without loss of generality, the state of the protocol
consists only of the inputs received by each party and the {\em transcript}
of the messages exchanged so far.%
\footnote{\label{foot:stateless}%
Since the parties are computationally unbounded, there is no need
to allow private randomness as part of the state; randomness for a party can
always be resampled at every round conditioned on the inputs, outputs and
messages in that party's view.}
We denote the final transcripts on the three links, after
executing protocol $\Pi$ on its specified input distribution by
$M^{\Pi}_{12}, M^{\Pi}_{23}$ and $M^{\Pi}_{31}$. When $\Pi$ is clear from
the context, we simply write $M_{12}$ etc.  We define $M_1=(M_{12},M_{31})$
as the transcripts that party 1 can see; $M_2$ and $M_3$ are defined
similarly. We define the view of the $i^\text{th}$ party, $V_i$ to consist
of $M_i$ and that party's inputs and outputs (if any).

It is easy to see that a (well-formed) protocol, along with an input
distribution, fully defines a joint distribution over all the inputs,
outputs and the joint transcripts on all the links.

\paragraph{Secure Computation.}
We consider three party computation functionalities, in which Alice and Bob (parties 1
and 2) receive as inputs the random variables $X\in{\X}$ and
$Y\in{\Y}$, respectively, and Charlie (party 3) produces an output
$Z\in{\Z}$ distributed according to a specified distribution
$p_{Z|XY}$. In particular, we can consider a {\em deterministic function
evaluation} functionality where $Z=f(X,Y)$ with probability 1, for some
function $f:{\X}\times{\Y}\rightarrow{\Z}$. The set
${\X}$, ${\Y}$ and ${\Z}$ are always finite. In secure computation, we shall consider the inputs to the computation to come from a distribution $p_{XY}$ over ${\X}\times{\Y}$.

A (perfectly) secure computation protocol $\Pi(p_{XY},p_{Z|XY})=(\Pi_1,\Pi_2,\Pi_3,\out3)$ for $(p_{XY},p_{Z|XY})$ 
is a protocol which satisfies the following conditions:
\begin{itemize}
\item Correctness: Output of Charlie, is distributed according to $p_{Z|X=x,Y=y}$,
where $x,y$ are the inputs to Alice and Bob
\item Privacy: The privacy condition corresponds to ``1-privacy'', wherein at
most one party is passively corrupt. Corresponding to security against
Alice, Bob and Charlie, respectively, we have the following three Markov chains.
$M_1 - X - (Y,Z)$, $M_2 - Y - (X,Z)$ and $M_3 - Z - (X,Y)$.
Equivalently, in terms of the views, $I(V_1;(Y,Z)|X)=I(V_2;(X,Z)|Y)=I(V_3;(X,Y)|Z)=0$.
\end{itemize}
Intuitively, the privacy condition guarantees that even if one party (say
Alice) is curious, and retains its view from the protocol (i.e., $M_1$),
this view reveals nothing more to it about the inputs and outputs of the
other parties (namely, $Y,Z$), than what its own inputs and outputs reveal
(as long as the other parties erase their own views). In other words, a
curious party may as well simulate a view for itself based on just its
inputs and outputs, rather than retain the actual view it obtained from the
protocol execution.

For simplicity, we prove all our results for {\em perfect security} as defined
above; this is also the setting for classical positive results like
that of \cite{BenorGoWi88}. But all our bounds do extend to the setting of
statistical security, as we shall show in the full version of this paper
(following \cite{WinklerWu10,PrabhakaranPr12} who extend similar results to
the statistical security case).%
\footnote{We remark that we do not know if our bounds extend to a 
relaxed security setting sometimes considered in the information theory
literature: there the error in security is only required to go to 0 as the
size of the input grows to infinity. Instead, we use the standard
cryptographic security requirement that for any fixed input length, the
error can be driven arbitrarily close to zero by choosing a large enough
security parameter.}
Also, the above security requirements are for an honest execution of the
protocol (corresponding to honest-but-curious or passive corruption of at
most one party). The lower bounds derived in this model typically continue to hold
for active corruption as well (since for many functionalities,
every protocol secure against active corruption is a protocol secure against
passive corruption), but in fact, in our setting (where 1 out of 3 parties is
corrupted), the functions we consider simply do not have secure protocols
against active corruption.


\paragraph{Communication Complexity and Entropy.} A standard approach to
lowerbounding the number of bits in a string is to lowerbound its entropy.
However, in an interactive setting, a party sees the messages in each
round, rather than just a concatenation of all the bits sent over the
entire protocol. In a setting where we allow variable length messages, this
would seem to allow communicating more bits of information than the length
of the transcript itself. But this allows the parties to learn when the
message transmitted in a round ends, implicitly inserting an end-of-message
marker into the bit stream. To account for this, one can require that the
message sent at every round is a codeword in a prefix-free code. (The code
itself can be dynamically determined based on previous messages exchanged
over the link.) As shown in
\Appendixref{communication-lowerbound-through-entropy-lowerbound},
with this requirement, the number of bits communicated in each link is
indeed lowerbounded by the entropy of the transcript in that link.

\paragraph{Normal Forms.} 
In \Appendixref{prelims}, we describe a normal form for a randomized function $p_{Z|XY}$ as well as for the pair $(p_{XY},p_{Z|XY})$, where $p_{XY}$ and $p_{Z|XY}$ are the input distribution and the function respectively. Essentially, these normal forms merge ``equivalent'' inputs and outputs. As argued there, it suffices to study the communication complexity of secure computation for functions $p_{Z|XY}$ and for pairs $(p_{XY},p_{Z|XY})$ in normal form.

%

\paragraph{Communication-Ideal Protocol.} We say that a protocol
$\Pi(p_{XY},p_{Z|XY})$ for
securely computing a randomized function $p_{Z|XY}$,
for a distribution $p_{XY}$ 
is {\em
communication-ideal} if for each $ij \in \{12,23,31\}$, 
\[ H(M^{\Pi}_{ij}) = \inf_{\Pi'(p_{XY},p_{Z|XY})} H(M^{\Pi'}_{ij}), \]
where the infimum is over all secure protocols for $p_{Z|XY}$ with the
same distribution $p_{XY}$.
That is, a communication-ideal protocol achieves the least entropy
possible for every link, simultaneously. We remark that it is not clear, {\em
a priori}, how to determine if a given function $p_{Z|XY}$ has a
communication-ideal protocol for a given distribution $p_{XY}$.

\subsection*{Common Information and Residual Information}
G\'acs and K\"orner~\cite{GacsKorner} introduced the notion of common
information to measure a certain aspect of correlation between two random
variables. The G\'acs-K\"orner common information of a pair of correlated
random variables $(U,V)$ can be defined as $H(U\sqcap V)$, where $U\sqcap V$
is a random variable with maximum entropy among all random variables $Q$ that
are determined both by $U$ and by $V$ (i.e., there are functions $f$ and $g$
such that $Q =f(U)=g(V)$). In \cite{PrabhakaranPr12}, the gap between mutual information and
common information was termed {\em residual information}:
$\RI(U;V):=I(U;V)-H(U\sqcap V)$.

In \cite{WolfWu08}, Wolf and Wullschleger identified (among other things) 
the following important {\em data processing inequality} for residual information.
\begin{lem}[\cite{WolfWu08}]\label{lem:monotone}
If $T,U,V,W$ are jointly distributed random variables such that the following two Markov chains hold: (i) $U-T-W$, and (ii) $T-W-V$, then
\[  \RI(T;W) \leq \RI((U,T);(V,W)).\]
\end{lem}
The Markov chain conditions above correspond to the requirement that
it is secure (against honest-but-curious adversaries) to require
a pair of parties holding the views $(U,T)$ and $(V,W)$, to produce outputs
$T,W$, respectively, because for the first party, the rest of its view, $U$,
can be simulated based on the output $T$, independent of the output $W$ (and
similarly, for the second party). The lemma states that under such a secure
transformation from views to outputs, the residual information can only
decrease.

In \cite{PrabhakaranPr12}, the following alternate definition of residual information
was given, which will be useful in lowerbounding conditional mutual information
terms.
\begin{align}
\RI(U;V)=\min_{\substack{Q:\exists f,g \text{ s.t. }\\Q=f(U)=g(V)}} I(U;V|Q). \label{eq:RI_minform}
\end{align}
The random variable $Q$ which achieves the minimum is, in fact, $U\sqcap V$.
Note that the residual information is always non-negative.

\section{Lower Bounds on Communication Complexity}\label{sec:lowerbounds}
This section is divided into three parts. In \Subsectionref{prelim_lbs}, we derive preliminary lower bounds for secure computation. In each of the subsequent subsections, we give different improvements of the lower bounds derived in \Subsectionref{prelim_lbs}.

\subsection{Preliminary Lower Bounds}\label{subsec:prelim_lbs}
We first state the following basic lemma for any protocol for secure computation. Similar results have appeared in the literature earlier (for instance, special cases of \Lemmaref{cutset} appear in \cite{DodisMi00,WinklerWu10,DPAllerton13}).

\Lemmaref{cutset} states the simple fact that, for $(p_{XY},p_{Z|XY})$ in normal form, the information about a party's input must flow out through the links she/he is part of, and the information about Charlie's output must flow in through the links he is part of. This crucially relies on the fact that Alice and Bob obtain no output, and Charlie has no input in our model.

\begin{lem}\label{lem:cutset}
Suppose $(p_{XY},p_{Z|XY})$ is in normal form. Then, in any secure protocol $\Pi(p_{XY},p_{Z|XY})$,  the cut isolating Alice from Bob and Charlie must reveal Alice's input $X$, i.e., $H(X| M_{12},M_{31}) =0$. Similarly, $H(Y|M_{12},M_{23})=0$ and $H(Z|M_{23},M_{31})=0$.
\end{lem}

A proof is given in \Appendixref{proofs}. We obtain a preliminary lower bound in \Theoremref{prelim_lbs} below by using the above lemma and the data-processing inequality for residual information in \Lemmaref{monotone}. Note that the assumption of $(p_{XY},p_{Z|XY})$ being in normal form below is without loss of generality (\Appendixref{prelims}).

\begin{thm}\label{thm:prelim_lbs}
Any secure protocol $\Pi(p_{XY},p_{Z|XY})$, where $(p_{XY},p_{Z|XY})$ is in normal form, should satisfy the following lower bounds on the entropy of the transcripts on each link.
\begin{align}
H(M_{23}) &\geq \max\{RI(X;Z), RI(X;Y)\} + H(Y,Z|X) \label{eq:prelim_lb_M23},\\
H(M_{31}) &\geq \max\{RI(Y;Z), RI(X;Y)\} + H(X,Z|Y) \label{eq:prelim_lb_M31},\\
H(M_{12}) &\geq \max\{RI(X;Z), RI(Y;Z)\} + H(X,Y|Z) \label{eq:prelim_lb_M12}.
\end{align}
\end{thm}
\begin{proof}
We shall prove \eqref{eq:prelim_lb_M23}. The other two inequalities 
follow symmetrically.
\begin{align}
H(M_{23}) &\geq \max\{H(M_{23}|M_{31}),H(M_{23}|M_{12})\} \notag \\
&= \max\{I(M_{23};M_{12}|M_{31}),I(M_{23};M_{31}|M_{12})\} + H(M_{23}|M_{12},M_{31}).
\label{eq:prelim_bound_M23}
\end{align}
Firstly, we can bound the last term of \eqref{eq:prelim_bound_M23} as follows (to
already get a na\"ive bound):
\begin{align*}
H(M_{23}|M_{12},M_{31}) &\stackrel{\text{(a)}}{=} H(M_{23},Y,Z|M_{12},M_{31},X)\\
&\geq H(Y,Z|M_{12},M_{31},X) \stackrel{\text{(b)}}{=} H(Y,Z | X),
\end{align*}
where (a) follows from \Lemmaref{cutset} and (b) follows from the privacy against Alice.
Next, we lower bound the first term inside $\max$ of
\eqref{eq:prelim_bound_M23} by $RI(X;Z)$ as follows. Firstly,
\begin{align}
\label{eq:condI-vs-RI}
I(M_{23};M_{12}|M_{31}) = I(M_{23},M_{31};M_{12},M_{31}|M_{31}) \geq RI(M_{23},M_{31};M_{12},M_{31}),
\end{align}
where the inequality follows from \eqref{eq:RI_minform}, the alternate
definition of residual information, by taking
$Q=M_{31}$.
Now, by privacy against Charlie, we have $(M_{23},M_{31})-Z-X$ and privacy against
Alice, we have $(M_{12},M_{31})-X-Z$. Applying \Lemmaref{monotone} with the
above markov chains, together with \Lemmaref{cutset}, we get
\[RI(M_{23},M_{31};M_{12},M_{31}) \geq RI(Z;X) = RI(X;Z).\]
Similarly, we can lower bound the second term inside $\max$ of \eqref{eq:prelim_bound_M23} by $RI(X;Y)$,
completing the proof.
\end{proof}

A consequence of \Lemmaref{cutset} is that the transcripts in a secure computation protocol form shares in a CMSS scheme for the same distribution $p_{XYZ}=p_{XY}p_{Z|XY}$; see \Appendixref{CMSS_sampling}. There we derive bounds on the sizes of these shares which, in fact, imply \Theoremref{prelim_lbs} (and \Theoremref{improved_prelim_lbs}). In the rest of the paper we will restrict our attention to $p_{XY}$ which have full support (and, without loss of generality, $p_{Z|XY}$ expressed in normal form). This will allow us to strengthen the preliminary bounds in \Theoremref{prelim_lbs}. Notice that such $(p_{XY}, p_{Z|XY})$ are in normal form and hence \Lemmaref{cutset} holds.


\subsection{Distribution Switching and Improved Lower Bounds - I}\label{subsec:improved_lbs}

To improve the bounds in \Theoremref{prelim_lbs}, we give a technique, {\em distribution switching}, which significantly improves the above bounds and leads to one of our main theorems.

The following lemma states that privacy requirements imply that the transcript $M_{12}$ generated by a secure protocol computing $p_{Z|XY}$ is independent of both the inputs. Moreover, if the function $p_{Z|XY}$ satisfies some additional constraints, then the other two transcripts also become independent of the inputs. For a distribution $p_{XY}$, a bipartite graph on vertex set $\X \cup \Y$ is said to be the {\em characteristic bipartite graph of $p_{XY}$}, if $x\in\X$ and $y\in\Y$ are connected whenever $p_{XY}(x,y)>0$. The proof of the following lemma is along the lines of a similar lemma in~\cite{DPAllerton13} and we prove it in \Appendixref{proofs} for completeness.

\begin{lem}\label{lem:XYZ_inde_M123}
Consider a function $p_{Z|XY}$ not necessarily in normal form. 
\begin{enumerate}
\item[1. ] Suppose that $p_{XY}$ is such that the characteristic bipartite graph of $p_{XY}$ is connected. Then, for any secure protocol $\Pi(p_{XY},p_{Z|XY})$, we have $I(X,Y,Z;M_{12})=0$.
\item[2. ] Suppose $p_{XY}$ has full support and $p_{Z|XY}$ satisfies the following condition:\\[0.15cm]
{\bf Condition 1.} There is no non-trivial partition $\mathcal{X} = \mathcal{X}_1 \cup \mathcal{X}_2$ (i.e., $\mathcal{X}_1 \cap \mathcal{X}_2 = \varnothing$ and neither $\mathcal{X}_1$ nor $\mathcal{X}_2$ is empty), such that if $\mathcal{Z}_k = \{z \in {\mathcal Z} : x \in \mathcal{X}_k, y \in \mathcal{Y}, p(z|x,y)>0\}, k=1,2$, their intersection $\mathcal{Z}_1 \cap \mathcal{Z}_2$ is empty. \\[0.15cm]
Then, for any secure protocol $\Pi(p_{XY},p_{Z|XY})$, we have $I(X,Y,Z;M_{31})=0$.
\item[3. ] Suppose $p_{XY}$ has full support and $p_{Z|XY}$ satisfies the following condition:\\[0.15cm]
{\bf Condition 2.} There is no non-trivial partition $\mathcal{Y} = \mathcal{Y}_1 \cup \mathcal{Y}_2$ such that if $\mathcal{Z}_k = \{z\in{\mathcal Z} : x \in \mathcal{X}, y \in \mathcal{Y}_k, p(z|x,y)>0\}, k=1,2$, their intersection $\mathcal{Z}_1 \cap \mathcal{Z}_2$ is empty.\\[0.15cm]
Then, for any secure protocol $\Pi(p_{XY},p_{Z|XY})$, we have $I(X,Y,Z;M_{23})=0$.
\end{enumerate}
\end{lem}
We point out that $p_{XY}$ having a connected characteristic bipartite graph is a weaker condition than $p_{XY}$ having full support.


\subsection*{Distribution Switching}

We will now strengthen the lower bounds from \Theoremref{prelim_lbs}. Specifically, we will argue that even if the protocol is allowed to depend on the input distribution (as we do here), privacy requirements will require that the lower bounds derived for when the distributions of the inputs are changed continue to hold for the original setting. 

We note that any secure protocol $\Pi(p_{XY},p_{Z|XY})$, where distribution $p_{XY}$ has full support, continues to be a secure protocol even if we switch the input distribution
to a different one $p_{X'Y'}$. This follows directly from examining the correctness and privacy conditions required for a protocol to be secure. 

\begin{thm}\label{thm:improved_prelim_lbs}
For any secure protocol $\Pi(p_{XY},p_{Z|XY})$, where $p_{Z|XY}$ is in normal form and $p_{XY}$ has full support, we have the following strengthening of~\eqref{eq:prelim_lb_M12}:
\begin{align}
H(M_{12}) &\geq \max\{\sup_{p_{X'Y'}}\left(RI(X';Z')+H(X',Y'|Z')\right), \sup_{p_{X'Y'}}\left(RI(Y';Z') + H(X',Y'|Z')\right)\}, \label{eq:improved_prelim_lb_M12}
\end{align}
where the supremizations are over $p_{X'Y'}$ having full support and the objective functions are evaluated using $p_{X'Y'Z'}(x,y,z)=p_{X'Y'}(x,y)p_{Z|XY}(z|x,y)$.\\[0.15cm]
If $p_{Z|XY}$ (in normal form) satisfies Condition~1 of \Lemmaref{XYZ_inde_M123}, then for any secure protocol $\Pi(p_{XY},p_{Z|XY})$, where $p_{XY}$ has full support, we have the following strengthening of~\eqref{eq:prelim_lb_M31}:
\begin{align}
H(M_{31}) &\geq \max\{\sup_{p_{X'Y'}}\left(RI(Y';Z')+H(X',Z'|Y')\right), \sup_{p_{X'Y'}}\left(RI(X';Y') + H(X',Z'|Y')\right)\}, \label{eq:improved_prelim_lb_M31}
\end{align}
where the supremizations are over $p_{X'Y'}$ having full support and the objective functions are evaluated using $p_{X'Y'Z'}(x,y,z)=p_{X'Y'}(x,y)p_{Z|XY}(z|x,y)$.\\[0.15cm]
If $p_{Z|XY}$ (in normal form) satisfies Condition~2 of \Lemmaref{XYZ_inde_M123}, then for any secure protocol $\Pi(p_{XY},p_{Z|XY})$, where $p_{XY}$ has full support, we have the following strengthening of~\eqref{eq:prelim_lb_M23}:
\begin{align}
H(M_{23}) &\geq \max\{\sup_{p_{X'Y'}}\left(RI(X';Z')+H(Y',Z'|X')\right), \sup_{p_{X'Y'}}\left(RI(X';Y') + H(Y',Z'|X')\right)\}, \label{eq:improved_prelim_lb_M23}
\end{align}
where the supremizations are over $p_{X'Y'}$ having full support and the objective functions are evaluated using $p_{X'Y'Z'}(x,y,z)=p_{X'Y'}(x,y)p_{Z|XY}(z|x,y)$.\\
\end{thm}
\begin{proof}
By \Lemmaref{XYZ_inde_M123}, it follows that the transcript $M_{12}$ of the protocol (under both the original and the switched input distributions) must remain independent of the input data $X,Y$. This allows us to argue using \Theoremref{prelim_lbs}, that 
\[H(M_{12}) \geq \max\{\sup_{p_{X'Y'}}\left(RI(X';Z')+H(X',Y'|Z')\right), \sup_{p_{X'Y'}}\left(RI(Y';Z') + H(X',Y'|Z')\right)\},\] where the supremizations are over $p_{X'Y'}$ having full support and the objective functions are evaluated using $p_{X'Y'Z'}(x,y,z)=p_{X'Y'}(x,y)p_{Z|XY}(z|x,y)$.

Similarly, if the function $p_{Z|XY}$ satisfies the condition 1 and 2 of \Lemmaref{XYZ_inde_M123}, we can show the other two bounds on $H(M_{31})$ and $H(M_{23})$ as well.
\end{proof}
In \Appendixref{CMSS_sampling}, we derive similar bounds for the size of the shares of a CMSS scheme.


\subsection{An Information Inequality for Protocols and Improved Lower Bounds - II}\label{subsec:improved_lbs2}
We can give a different improvement to \Theoremref{prelim_lbs} by exploiting the fact that in a protocol, transcripts are generated by the parties interactively, rather than by an omniscient dealer. Towards this, we derive an information inequality relating the transcripts on different links in general 3-party protocols, in which parties do not share any common or correlated randomness or correlated inputs at the beginning of the protocol. Note that our model for protocols does indeed satisfy these conditions when the inputs are independent of each other.
\begin{lem}
\label{lem:infoineq}
In any well-formed 3-party protocol, if the inputs to the parties are
independent of each other, then, for $\{\alpha,\beta,\gamma\}=\{1,2,3\}$,
\[ I(M_{\gamma\alpha};M_{\beta\gamma}) \geq
I(M_{\gamma\alpha};M_{\beta\gamma}|M_{\alpha\beta}).\]
\end{lem}
We prove the lemma in \Appendixref{proofs}. We further
note that, as in \eqref{eq:condI-vs-RI},
$I(M_{\gamma\alpha};M_{\beta\gamma}|M_{\alpha\beta}) \ge$ \\ $RI(M_{\gamma\alpha}M_{\alpha\beta};M_{\beta\gamma}M_{\alpha\beta})$.
Hence, if the inputs are independent of each other,
\begin{align}
I(M_{\gamma\alpha};M_{\beta\gamma})
\ge I(M_{\gamma\alpha};M_{\beta\gamma}|M_{\alpha\beta})
\ge RI(M_{\gamma\alpha}M_{\alpha\beta};M_{\beta\gamma}M_{\alpha\beta}).
\label{eq:infoineq-corol}
\end{align}

This inequality provides us with a means to exploit the protocol structure
behind transcripts. Below, \Theoremref{intermediate_lbs} (specifically,
\eqref{eq:intermediate_lb_M12}) shows that the term $\max\{RI(X;Z),
RI(Y;Z)\}$ in \eqref{eq:prelim_lb_M12} can be replaced by $RI(X;Z)+
RI(Y;Z)$. Note that \Theoremref{intermediate_lbs} is stated and proven for
independent inputs, that is, $p_{XY}=p_Xp_Y$.
In~\Appendixref{dependent_inputs_lbs} we show that (using ideas of
distribution switching from \Subsectionref{improved_lbs}) it implies lower bounds for
dependent inputs $p_{XY}$ with full support as well. However, this extension
is not required in the sequel where we derive our main theorems.

\begin{thm}\label{thm:intermediate_lbs} Any secure protocol $\Pi(p_Xp_Y,p_{Z|XY})$, where $p_{Z|XY}$ is in normal form and $p_X$, $p_Y$ have full support, should satisfy the following lower bounds on the entropy of the transcripts on each link.
\begin{align}
H(M_{23}) &\geq RI(X;Z) + H(Y,Z|X) \label{eq:intermediate_lb_M23}\\
H(M_{31}) &\geq RI(Y;Z) + H(X,Z|Y) \label{eq:intermediate_lb_M31}\\
H(M_{12}) &\geq RI(X;Z) + RI(Y;Z) + H(X,Y|Z) \label{eq:intermediate_lb_M12}
\end{align}
\end{thm}
\begin{proof}
Firstly, note that 
\eqref{eq:intermediate_lb_M23} and \eqref{eq:intermediate_lb_M31}
follow from 
\eqref{eq:prelim_lb_M23} and \eqref{eq:prelim_lb_M31}. To prove
\eqref{eq:intermediate_lb_M12}, we have
\begin{align*}
H(M_{12}) &= I(M_{12};M_{23}) + H(M_{12}|M_{23}) \\
&= I(M_{12};M_{23}) + I(M_{12};M_{31}|M_{23}) + H(M_{12}|M_{23},M_{31})\\ 
&\ge RI(X;Z) + RI(Y;Z) + H(X,Y | Z),
\end{align*}
where the last inequality used $H(M_{12}|M_{23},M_{31}) \ge H(X,Y | Z)$ and
$I(M_{12};M_{31}|M_{23}) \ge RI(Y;Z)$ (both as in the proof of
\Theoremref{prelim_lbs}) as well as
$I(M_{12};M_{23}) \geq RI(X;Z)$
(by \eqref{eq:infoineq-corol}, 
which applies since we assume independent inputs).
\end{proof}

\noindent In \Appendixref{CMSS_sampling} we show that the above proof can be extended to derive lower bounds for secure sampling.

We can improve the bounds in \Theoremref{intermediate_lbs} using distribution switching which leads to our main theorems. The following lemma states that if the inputs $X$ and $Y$ are independent, then privacy requirements imply that certain transcripts generated by a secure protocol computing $p_{Z|XY}$, are independent of certain data. More precisely, we show the following.
\begin{lem}\label{lem:XY_inde_M23M31}
For any secure protocol $\Pi(p_Xp_Y,p_{Z|XY})$, where $p_{Z|XY}$ may not be in normal form, should satisfy $I(X;M_{23})=I(Y;M_{31})=0$.
\end{lem}
\begin{proof}
$I(X;M_{23})=0$ follows from $I(X;M_{23}) \leq I(X;M_{23},Y)= $  $I(X;Y) +$  $I(X;M_{23}|Y)=0$, where the last equality follows from independence of $X$ and $Y$ and privacy against Bob. Similarly we can show $I(Y;M_{31})=0$.
\end{proof}
If the inputs $X$ and $Y$ are independent, i.e., $p_{XY}=p_Xp_Y$, then the transcripts $M_{23}$ and $M_{31}$ of the protocol (under both the original and the switched input distributions) must remain independent of $X$ and $Y$ respectively (by \Lemmaref{XY_inde_M23M31}). Note that for the independence of $M_{12}$ and $(X,Y)$, we do not need the independence of inputs. Rather, since we used \Lemmaref{infoineq} (which requires independence of inputs) to get~\eqref{eq:intermediate_lb_M12}, we are forced to consider independent inputs if we use the bound on $H(M_{12})$ in~\eqref{eq:intermediate_lb_M12}. This allows us to argue using \Theoremref{intermediate_lbs}, that
\[ H(M_{12}) \geq \sup_{p_{X'}p_{Y'}} \RI(X';Z') + \RI(Y';Z') + H(X',Y'|Z'),\]
where the supremum is over $p_{X'},p_{Y'}$ which have full support and
the terms in the right hand side are evaluated under the joint distribution
\[p_{X',Y',Z'}(x,y,z) = p_{X'}(x)p_{Y'}(y)p_{Z|X,Y}(z|x,y).\]
Similarly, from \Lemmaref{XY_inde_M23M31}, we know that $M_{31}$ is independent of
$Y$, and $M_{23}$ is independent of $X$. Hence,
\begin{align}
H(M_{31}) &\geq \sup_{p_{Y'}} \RI(Y';Z') + H(X,Z'|Y'),
\label{eq:distrswitchintermediate1}\\
H(M_{23}) &\geq \sup_{p_{X'}} \RI(X';Z') + H(Y,Z'|X'),
\end{align}
where the right hand side of \eqref{eq:distrswitchintermediate1} is
evaluated under $p_{X,Y',Z'}(x,y,z) = p_X(x) p_{Y'}(y) p_{Z|X,Y}(z|x,y)$.
Similarly, for the bound on $H(M_{23})$.

In fact, we can show an even stronger bound than above by a more careful application of distribution switching. This leads us to second of our three main lower bound theorems, which is proved in \Appendixref{main_thms_proofs}.

\noindent {\em Remark:} The above discussion of distribution switching used
\Lemmaref{XY_inde_M23M31}, which holds for independent inputs, that is,
$p_{XY}=p_Xp_Y$. However, as we argue in
\Appendixref{dependent_inputs_lbs}, the resulting lower bounds also hold for
dependent inputs $p_{XY}$ with full support; we state the theorem for this
general case.
\begin{thm} \label{thm:lowerbound}
The following communication complexity bounds hold for any secure protocol $\Pi(p_{XY},p_{Z|XY})$, where $p_{Z|X,Y}$ is in normal form and $p_{XY}$ has full support:
\begin{align}
H(M_{23}) &\geq \left(\sup_{p_{X'}}RI(X';Z')\right) + \left(\sup_{p_{X''}}H(Y,Z''|X'') \right), \label{eq:M23_NoCond}\\
H(M_{31}) &\geq \left(\sup_{p_{Y'}}RI(Y';Z')\right) + \left(\sup_{p_{Y''}}H(X,Z''|Y'') \right), \label{eq:M31_NoCond}\\
H(M_{12}) &\geq \max \left\{ \begin{array}{l} \sup_{p_{X'}} \left(\sup_{p_{Y'}} RI(Y';Z')\right) + \left(\sup_{p_{Y''}} RI(X';Z'') + H(X',Y''|Z'')\right),\\ 
\sup_{p_{Y'}} \left(\sup_{p_{X'}} RI(X';Z')\right) + \left(\sup_{p_{X''}} RI(Y';Z'') + H(X'',Y'|Z'')\right)
\end{array}\right\}, \label{eq:M12_bound}
\end{align}
where the supremizations are over distributions $p_{X'}, p_{X''}, p_{Y'}, p_{Y''}$ having full support. The terms in the right hand side of \eqref{eq:M23_NoCond} are evaluated using the distribution $p_Y$ of the data $Y$ of Bob, i.e., 
\begin{align*}
p_{X',Y,Z'}(x,y,z)&=p_{X'}(x)p_{Y}(y)p_{Z|X,Y}(z|x,y),\\
p_{X'',Y,Z''}(x,y,z)&=p_{X''}(x)p_{Y}(y)p_{Z|X,Y}(z|x,y).\end{align*}
Similarly, the terms in \eqref{eq:M31_NoCond} are evaluated using the the distribution $p_X$ of the data $X$ of Alice. The lower bound in \eqref{eq:M12_bound} does not depend on the distributions $p_X$ and $p_Y$ of the data. The terms on the top row of~\eqref{eq:M12_bound}, for instance, are evaluated using
\begin{align*}
p_{X',Y',Z'}(x,y,z)&=p_{X'}(x)p_{Y'}(y)p_{Z|X,Y}(z|x,y),\\
p_{X',Y'',Z''}(x,y,z)&=p_{X'}(x)p_{Y''}(y)p_{Z|X,Y}(z|x,y).
\end{align*}
\end{thm}

\noindent When the function satisfies certain additional constraints, we can strengthen the lower bounds on the $H(M_{23})$ and $H(M_{31})$ as shown in the following theorem which is proved in \Appendixref{main_thms_proofs}.

\begin{thm} \label{thm:conditionallowerbounds}
Consider any secure protocol $\Pi(p_{XY},p_{Z|XY})$, where $p_{XY}$ has full support and $p_{Z|XY}$ is in normal form.
\begin{enumerate}
\item Suppose the function $p_{Z|XY}$ satisfies Condition~1 of \Lemmaref{XYZ_inde_M123}, that is, there is no non-trivial partition $\mathcal{X} = \mathcal{X}_1 \cup \mathcal{X}_2$ (i.e., $\mathcal{X}_1 \cap \mathcal{X}_2 = \varnothing$ and neither $\mathcal{X}_1$ nor $\mathcal{X}_2$ is empty), such that if
$\mathcal{Z}_k = \{z \in {\mathcal Z} : x \in \mathcal{X}_k, y \in \mathcal{Y}, p(z|x,y)>0\}, k=1,2$, their intersection $\mathcal{Z}_1 \cap \mathcal{Z}_2$ is empty. Then, we have the following strengthening of \eqref{eq:M31_NoCond}.
\begin{align}
H(M_{31}) \geq \sup_{p_{X'}} \left(\left(\sup_{p_{Y'}} RI(Y';Z') \right) +
\left( \sup_{p_{Y''}} H(X',Z''|Y'') \right) \right), \label{eq:M31_WithCond}
\end{align}
where the suprimizations are over distributions $p_{X'},p_{Y'},p_{Y''}$ having full
support and the terms in the right hand side are evaluated using the
distribution
\[p_{X'Y'Z'Y''Z''}(x',y',z',y'',z'')=p_{X'}(x')p_{Y'}(y')p_{Z|XY}(z'|x',y')p_{Y''}(y'')p_{Z|XY}(z''|x',y'').\]

\item Suppose the function $p_{Z|XY}$ satisfies Condition~2 of \Lemmaref{XYZ_inde_M123}, that is, there is no non-trivial partition $\mathcal{Y} = \mathcal{Y}_1 \cup \mathcal{Y}_2$ such that if $\mathcal{Z}_k = \{z\in{\mathcal Z} : x \in \mathcal{X}, y \in \mathcal{Y}_k, p(z|x,y)>0\}, k=1,2$, their intersection $\mathcal{Z}_1 \cap \mathcal{Z}_2$ is empty. Then, we have the following strengthening of \eqref{eq:M23_NoCond}.
\begin{align}
H(M_{23}) \geq \sup_{p_{Y'}} \left( \left( \sup_{p_{X'}} RI(X';Z') \right) +
\left( \sup_{p_{X''}} H(Y',Z''|X'') \right) \right), \label{eq:M23_WithCond}
\end{align}
where the supremizations are over distributions $p_{X'},p_{X''},p_{Y'}$ having full
support and the terms in the right hand side are evaluated using the
distribution
\[p_{X'Y'Z'X''Z''}(x',y',z',x'',z'')=p_{X'}(x')p_{Y'}(y')p_{Z|XY}(z'|x',y')p_{X''}(x'')p_{Z|XY}(z''|x'',y').\]
\end{enumerate}
\end{thm}

Note that in \Theoremref{improved_prelim_lbs}, \Theoremref{lowerbound} and
\Theoremref{conditionallowerbounds}, any choice of $p_{X'Y'}$,
$p_{X'},p_{X''}$, $p_{Y'},p_{Y''}$ (with full support) will yield a lower
bound. For a given function, while all choices do yield valid lower bounds,
one is often able to obtain the {\em best} lower bound analytically (as in
\Theoremref{ot}, where it is seen to be the best as it matches an upper
bound) or numerically (as in \Theoremref{and}).

To summarize, for any secure computation problem $(p_{XY},p_{Z|XY})$, expressed in the normal form, \Theoremref{prelim_lbs} gives lower bounds on entropies of transcripts on all three links. If $p_{XY}$ has full support and $p_{Z|XY}$ is in normal form, then for $H(M_{31})$, our best lower bound is the larger of~\eqref{eq:prelim_lb_M31} and~\eqref{eq:M31_NoCond}; for $H(M_{23})$, it is the larger of~\eqref{eq:prelim_lb_M23} and~\eqref{eq:M23_NoCond} and for $H(M_{12})$, it is the larger of~\eqref{eq:improved_prelim_lb_M12} and~\eqref{eq:M12_bound}. In addition, if $p_{Z|XY}$ satisfies condition~1 of \Lemmaref{XYZ_inde_M123}, then for $H(M_{31})$, our best lower bound is the larger of~\eqref{eq:improved_prelim_lb_M31} and~\eqref{eq:M31_WithCond}; if $p_{Z|XY}$ satisfies condition~2 of \Lemmaref{XYZ_inde_M123}, then for $H(M_{23})$, our best lower bound is the larger of~\eqref{eq:improved_prelim_lb_M23} and~\eqref{eq:M23_WithCond}.


\section{Lower Bounds on Randomness}\label{sec:randomness}
In this section, we provide lower bounds on the amount of randomness required in secure computation protocols. Although our focus in this paper is not to prove lower bounds on the amount of randomness, it turns out that we may apply the above lower bounds on communication to derive bounds on the amount of randomness required. We show in \Sectionref{examples} that they give tight bounds on randomness required for the specific functions we analyse.
\begin{defn}
The randomness required to securely compute a function $p_{Z|XY}$ for input distribution $p_{XY}$ is defined as \[\rho(p_{XY},p_{Z|XY})=\inf_{\Pi(p_{XY},p_{Z|XY})} H(V_1,V_2,V_3|X,Y),\] where $V_i$ is the view of party-$i$ at the end of the protocol and infimum is over all secure protocols for $(p_{XY},p_{Z|XY})$.
\end{defn}
\noindent We can simplify the entropy term in the above definition as follows.
\begin{align*}
H(V_1,V_2,V_3|X,Y) &\geq H(M_{31},M_{12},M_{23}|X,Y)\\
&\geq \max\{H(M_{31}|X,Y), H(M_{12}|X,Y), H(M_{23}|X,Y)\}.
\end{align*}
Since the above inequalities are true for any secure protocol $\Pi(p_{XY},p_{Z|XY})$, we have
\begin{align}
\rho(p_{XY},p_{Z|XY}) \geq \max\{H(M_{31}|X,Y), H(M_{12}|X,Y), H(M_{23}|X,Y)\}.\label{eq:randomness}
\end{align}
\begin{thm}\label{thm:randomness}
Consider any secure protocol $\Pi(p_{XY},p_{Z|XY})$.
\begin{enumerate}
\item If the characteristic bipartite graph of $p_{XY}$ is connected, then $\rho(p_{XY},p_{Z|XY})\geq H(M_{12})$.
\item If $p_{XY}$ has full support and $p_{Z|XY}$ satisfies Condition~1 of \Lemmaref{XYZ_inde_M123}, then $\rho(p_{XY},p_{Z|XY})\geq H(M_{31})$.
\item If $p_{XY}$ has full support and $p_{Z|XY}$ satisfies Condition~2 of \Lemmaref{XYZ_inde_M123}, then $\rho(p_{XY},p_{Z|XY})\geq H(M_{23})$.
\end{enumerate}
\end{thm}
\begin{proof}
By \Lemmaref{XYZ_inde_M123}, we have the following:
\begin{enumerate}
\item If the characteristic bipartite graph of $p_{XY}$ is connected, then $H(M_{12}|X,Y)=H(M_{12})$. This, together with~\eqref{eq:randomness} implies $\rho(p_{XY},p_{Z|XY})\geq H(M_{12})$.
\item If $p_{XY}$ has full support and $p_{Z|XY}$ satisfies Condition~1, then $H(M_{31}|X,Y)=H(M_{31})$. This, together with~\eqref{eq:randomness} implies $\rho(p_{XY},p_{Z|XY})\geq H(M_{31})$.
\item If $p_{XY}$ has full support and $p_{Z|XY}$ satisfies Condition~2, then $H(M_{23}|X,Y)=H(M_{23})$. This, together with~\eqref{eq:randomness} implies $\rho(p_{XY},p_{Z|XY})\geq H(M_{23})$.
\end{enumerate}
\end{proof}

Hence, we can apply the lower bounds developed in \Sectionref{lowerbounds} to obtain lower bounds on randomness. If $p_{XY}$ has full support and $p_{Z|XY}$ is in normal form, then for $H(M_{12})$, our best lower bound is the larger of~\eqref{eq:improved_prelim_lb_M12} and~\eqref{eq:M12_bound}. In addition, if $p_{Z|XY}$ satisfies condition~1 of \Lemmaref{XYZ_inde_M123}, then for $H(M_{31})$, our best lower bound is the larger of~\eqref{eq:improved_prelim_lb_M31} and~\eqref{eq:M31_WithCond}; if $p_{Z|XY}$ satisfies condition~2 of \Lemmaref{XYZ_inde_M123}, then for $H(M_{23})$, our best lower bound is the larger of~\eqref{eq:improved_prelim_lb_M23} and~\eqref{eq:M23_WithCond}.

We call a protocol {\em randomness-optimal}, if the total number of random bits used by the protocol is optimal. In all the examples we consider in the next section, the amount of randomness does not depend on the input distribution $p_{XY}$ as long as they have full support, so, instead of writing $\rho(p_{XY},p_{Z|XY})$, we simply write $\rho(p_{Z|XY})$.

\section{Application to Specific Functions}\label{sec:examples}
%
%
%
%
%

In this section we consider a few important examples, and apply our generic
lower bounds from above to these examples, to obtain interesting results.
While many of these results are natural to conjecture, they are not easy to prove
(see, for instance,~\Footnoteref{addition}). 

\paragraph{Optimality of the FKN Protocol.} Feige et al.\ \cite{FeigeKiNa94}
provided a generic (non-interactive) secure computation protocol for all
3-party functions in our model. This protocol uses a straight-forward (but
``inefficient'') reduction from an arbitrary function to a variant of the oblivious transfer problem, which we shall call the remote OT function (defined below), and then gives a simple protocol for this new function. 
While the resulting protocol is inefficient for most functions, one could
ask whether the protocol that \cite{FeigeKiNa94} used for {\sc remote OT}
itself is optimal. We use our lower bounds from above to show that indeed,
this is the case. 

The {\sc remote $\binom{m}{1}$-OT$^n_2$} function, is defined as follows:
Alice's input $X=(X_0,X_1,\hdots,X_{m-1})$ is made up of $m$ bit-strings
each of length $n$, and Bob has an input $Y\in \{0,1,\hdots,m-1\}$. Charlie
wants to compute $Z=f(X,Y)=X_Y$.  \Figureref{ot} in \Appendixref{ot} gives
the simple protocol for this function from \cite{FeigeKiNa94}
(rephrased as a protocol in our model). It requires $nm$ bits to be
exchanged over the Alice-Charlie (31) link, $n+\log m$  bits over the
Bob-Charlie (23) link and $nm+\log m$ bits over the Alice-Bob (12) link. The total number of random bits used in the protocol is $nm+\log m$. In \Appendixref{ot}, we prove the following theorem, which shows that this
protocol is optimal and in fact, a {\em communication-ideal} protocol. We also prove that this protocol is randomness-optimal.
\begin{thm}\label{thm:ot}
Any secure protocol $\Pi(p_{XY},\text{\sc remote-ot})$ for computing {\sc remote $\binom{m}{1}$-OT$^n_2$} for inputs $X$ and $Y$ where $p_{XY}$ has full support, must satisfy
\begin{align*}
H(M_{31}) \geq nm, \quad H(M_{23}) &\geq n+\log m, \;\text{ and}\quad H(M_{12})\geq nm+\log m, \\
\rho(\text{\sc remote-ot}) &\geq nm+\log m.
\end{align*}
\end{thm}

\paragraph{More Functions with Communication-Ideal and Randomness-Optimal Protocols.} 
{\sc group-add}, addition in any group has a communication-ideal and randomness-optimal protocol, for any input
distribution with full support (see \Appendixref{group}).  As mentioned in
\Footnoteref{addition}, this is easy to see for the uniform distribution,
and using distribution switching, we can see that the same holds as long as
the input distribution has full support. A more interesting example, is a
function called {\sc controlled-erasure} that was studied in
\cite{DPAllerton13}. We resolve the communication complexity of this secure
computation problem fully, by showing that the protocol for this function
from \cite{DPAllerton13} is in fact communication-ideal as well as randomness-optimal, again, for every
input distribution with full support.

\paragraph{Separating Secure and Insecure Computation.}
A basic question of secure computation is whether it needs more bits to be
communicated than the input-size itself (which suffices for insecure
computation). While natural to expect, it is not easy to prove this.
In their restricted model, \cite{FeigeKiNa94}
showed a non-explicit result, that for securely computing {\em most} Boolean functions on the domain
$\{0,1\}^n \times \{0,1\}^n$, Charlie is required to receive at least $3n-4$
bits, which is significantly more than the $2n$ bits sufficient for insecure
computation. 

{\sc remote $\binom{2}{1}$-OT$^n_2$} from above already gives us an explicit example of a
function where this is true: the total input size is $2n+1$, but the
communication is at least $H(M_{31})+H(M_{23}) \ge 3n+1$.
To present an
easy comparison to the lower bound of \cite{FeigeKiNa94}, we can consider a
symmetrized variant of {\sc remote $\binom{2}{1}$-OT$^n_2$}, in which two instances of {\sc
remote $\binom{2}{1}$-OT$^n_2$} are combined, one in each direction. More specifically,
$X=(A_0,A_1,a)$ where $A_0,A_1$ are of length $(n-1)/2$ (for an odd $n$) and
$a$ is a single bit; similarly $Y=(B_0,B_1,b)$; the output of the function
is defined as an $(n-1)$ bit string $f(X,Y)=(A_b,B_a)$. Considering (say)
the uniform input distribution over $X,Y$, the bounds for {\sc remote $\binom{2}{1}$-OT$^n_2$}
add up to give us $H(M_{31}) \ge 3(n-1)/2+1$ and $H(M_{23}) \ge 3(n-1)/2+1$,
so that the communication with Charlie is lowerbounded by
$H(M_{31})+H(M_{23}) \ge 3n-1$. 

This compares favourably with the bound of \cite{FeigeKiNa94} in many ways:
our lower bound holds even in a model that allows interaction; in particular,
this makes the gap between insecure computation ($n-1$ bits in our case,
$2n$ bits for \cite{FeigeKiNa94}) and secure computation (about $3n$ bits
for both) somewhat larger. More importantly, our lower bound is explicit (and
tight for the specific function we use), whereas that of \cite{FeigeKiNa94}
is existential. However, our bound does not subsume that of
\cite{FeigeKiNa94}, who considered {\em Boolean} functions. Our results do
not yield a bound significantly larger than the input size, when the output
is a single bit. It appears that this regime is more amenable to
combinatorial arguments, as pursued in \cite{FeigeKiNa94}, rather than
information theoretic arguments. Finally, for the case of {\em random} Boolean
functions, it is plausible that the actual communication cost is exponential
in the input length, but none of the current techniques are capable of
delivering such a result. We leave it as a fascinating open problem to
obtain tight bounds in this regime, possibly by combining combinatorial and
information-theoretic approaches.

\paragraph{Separating Secure Computation and Secret Sharing.}\label{pg:and-example}
 Another
natural separation one expects is between the amount of communication needed
when the views (or transcripts) are generated by a secure computation
protocol, versus when they are generated by an omniscient ``dealer'' so that
the security requirements are met. As mentioned before, the latter setting
corresponds to the share sizes in a CMSS scheme. Again, while such a
separation is expected, it is not very easy to establish this, especially
for explicit examples. It requires us to establish a strong lower
bound for the secure computation problem as well as provide a CMSS
scheme that is better. None of the examples considered above yield this
separation.

We establish the separation using the 3-party {\sc and} function, defined as
follows: Alice has an input bit $X$, Bob has an input bit $Y$ and Charlie
should obtain $Z=f(X,Y)=X \land Y$. There is a CMSS scheme that achieves
$\log(3)\le 1.6$ bits of entropy for all three shares $M_{12},M_{23}$ and $M_{31}$
(see \Theoremref{gap_CMSS}). However, the following lower bounds, proven in
\Appendixref{and} using numerical optimization, shows that in a
secure computation protocol, $H(M_{12})$ should be strictly larger than this.
\begin{thm}\label{thm:and}
Any secure protocol $\Pi(p_{XY},\text{\sc and})$ for computing {\sc and} for inputs $X$ and $Y$ where $p_{XY}$ has full support over $\{0,1\}^n\times\{0,1\}^n$, must satisfy
\begin{align*}
H(M_{31}) \geq n\log(3), \quad H(M_{23}) &\geq n\log(3), \;\;\text{ and}\quad H(M_{12})\geq n(1.826),\\
\rho(\text{\sc and}) &\geq n(1.826).
\end{align*}
\end{thm}
The best known protocol for {\sc and} (which resembles the CMSS scheme above, and first appeared in
\cite{FeigeKiNa94}) achieves $H(M_{12}) = 1+\log(3), H(M_{23}) = H(M_{31}) = \log(3)$ (see \Appendixref{and}). Our lower bounds on $H(M_{31})$ and $H(M_{23})$ match with the protocol requirements on these links. The bound on $H(M_{12})$ is not known to be tight. The protocol given in \Appendixref{and} requires $1+\log(3)$ random bits and we prove a lower bound of 1.826.

%

\paragraph{Open Problems.} We close with a brief list of concrete open
problems from this work. For secure computation of
{\sc and} and {\sc sum} (see \Appendixref{sum} for {\sc sum}) there is a gap between the best known upper bound and our lower
bound for $M_{12}$ link.
These specific examples point to challenges in obtaining good lower bounds. Another important problem is to find an explicit example for a {\em Boolean} function in which the communication to Charlie
must be significantly larger than the total input size. Note that
\cite{FeigeKiNa94} gave an existential result (in their restricted model)
and the explicit example in this work does not have Boolean output. The
case of random Boolean functions, where communication being exponential in
the input length is a plausible, but unproven result, was already mentioned.

\bibliographystyle{alpha}
\bibliography{crypto,bib}

\appendix
\section{Preliminaries: More Details}
\label{app:prelims}

\paragraph{Protocol.}
In an execution of the protocol, (a subset of) the parties receive inputs
from the specified distribution, and exchange messages over private,
point-to-point links.  Without loss of generality, the state of the protocol
consists only of the inputs received by each party and the bits transmitted
on each link, so far. This is because, as the parties are computationally unbounded, private randomness 
for a party can always be resampled at every round conditioned on the inputs, outputs and
messages in that party's view.
Each $\Pi_i$ specifies a distribution over $\{0,1\}^* \times \{0,1\}^*$
(corresponding to the messages to be transmitted on the two links party $i$
is connected to) conditioned on the inputs of party $i$ and the messages in
the $i$-th party's links so far. For a protocol to be {\em well-formed}, we
require that the message sent by one party to another at any round is a
codeword in a prefix-free code such that the code itself is determined by the messages exchanged
between the two parties in previous rounds;%
\footnote{Without such a restriction, we implicitly allow ``end-of-message''
markers, and the rest of the communication could have fewer bits than the
entropy of the transcript annotated with rounds.}
also, we require that with probability 1, the protocol terminates --- i.e.,
it reaches a state when all $\Pi_i$ output empty strings (it is not
important for our lower bounds that the parties realize when this happens).

The final transcript in each of the three links consists of two
strings, obtained by simply concatenating all the bits sent in each
direction on that link, by the time the protocol has terminated.

\paragraph{A Normal Form for Functionality $p_{Z|XY}$.} 
For a randomized function $p_{Z|XY}$, define the relation  $x\equiv x'$ for $x,x'\in\X$ to hold
if $\forall y\in\Y, z\in\Z$, $p(z|x,y) = p(z|x',y)$; similarly $y\equiv y'$
is defined. We also define $z\equiv z'$ if there exists a constant $c$
such that $\forall x\in\X, y\in\Y$, $p(z|x,y)=c\cdot p(z'|x,y)$. We say that
$p_{Z|XY}$ is in the normal form if
$x\equiv x' \Rightarrow x=x'$,
$y\equiv y' \Rightarrow y=y'$ and
$z\equiv z' \Rightarrow z=z'$. 
 
It is easy to see that one can transform any randomized function $p_{Z|XY}$ to one
in normal form $p_{Z'|X'Y'}$, with possibly smaller alphabets, so that any
secure computation protocol for the former can be transformed to one for the
latter with the same communication costs, and vice versa.  (To define $X'$,
\X is modified by replacing all
$x$ in an equivalence class of $\equiv$ with a single representative; $Y'$ and
$Z'$ are defined similarly. The modification to the protocol, in either
direction, is for each party to locally map $X$ to $X'$ etc., or vice
versa; notice that the $Z'$ to $Z$ map is potentially randomized.)
Hence it is enough to study the communication complexity of securely
computing functions in the normal form.

\paragraph{A Normal Form for $(p_{XY},p_{Z|XY})$.}
We define a normal form for the pair $(p_{XY},p_{Z|XY})$, where $p_{XY}$ is the input distribution and the randomized function is $p_{Z|XY}$ as follows:
\begin{defn}
For a pair $(p_{XY},p_{Z|XY})$, define the relations $x\cong x'$, $y\cong y'$ and $z\cong z'$ as follows.
\begin{enumerate}
\item Take any $x,x'\in\X$ and define $\mathcal{S}_{x,x'}=\{y\in\Y : p_{XY}(x,y)>0, p_{XY}(x',y)>0\}$. We say that $x\cong x'$, if $\forall y\in\mathcal{S}_{x,x'}$ and $z\in\Z$, we have $p_{Z|XY}(z|x,y)=p_{Z|XY}(z|x',y)$.
\item Take any $y,y'\in\Y$ and define $\mathcal{S}_{y,y'}=\{x\in\X : p_{XY}(x,y)>0, p_{XY}(x,y')>0\}$. We say that $y\cong y'$, if $\forall x\in\mathcal{S}_{y,y'}$ and $z\in\Z$, we have $p_{Z|XY}(z|x,y)=p_{Z|XY}(z|x,y')$.
\item Take $z,z'\in\Z$ and define $\mathcal{S}=\{(x,y) : p_{XY}(x,y)>0\}$. We say that $z\cong z'$, if for some constant $c\geq 0$ and $\forall (x,y)\in\mathcal{S}$, we have $p_{Z|XY}(z|x,y)=c\cdot p_{Z|XY}(z'|x,y)$.
\end{enumerate}
A pair $(p_{XY},p_{Z|XY})$ is said to be in normal form if $x\cong x'\Rightarrow x=x'$, $y\cong y'\Rightarrow y=y'$, and $z\cong z'\Rightarrow z=z'$.
\end{defn}
We can assume without loss of generality that $(p_{XY},p_{Z|XY})$ is in normal form. Otherwise, suppose $x,x'\in\X$, where $x\neq x'$ and $x\cong x'$. In this case we can safely merge $x$ and $x'$ into a single $x^*$ without affecting anything. Now, $\forall (y,z)\in\Y\times\Z$, define $p_{XY}(x^*,y)=p_{XY}(x,y)+p_{XY}(x',y)$ and 
\[
 p_{Z|XY}(z|x^*,y) =
  \begin{cases}
   p_{Z|XY}(z|x,y) & \text{if } y \in \mathcal{S}_{x,x'}, \\
   p_{Z|XY}(z|x,y) & \text{if } p_{XY}(x,y)\geq 0 \text{ and } p_{XY}(x',y)=0, \\
   p_{Z|XY}(z|x',y) & \text{if } p_{XY}(x,y)=0 \text{ and } p_{XY}(x',y)>0, \\
  \end{cases}
\]
which gives $p_{XYZ}(x^*,y,z)=p_{XYZ}(x,y,z)+p_{XYZ}(x',y,z)$. Similarly we can merge equivalent $y$'s. For $z\neq z'$ with $z\cong z'$, we can merge $z$ and $z'$ into a single $z^*$ by defining $p_{Z|XY}(z^*|x,y)=p_{Z|XY}(z|x,y)+p_{Z|XY}(z'|x,y)$, which gives  $p_{XYZ}(x,y,z^*)=p_{XYZ}(x,y,z)+p_{XYZ}(x,y,z')$.

It is easy to see that one can transform any pair $(p_{XY},p_{Z|XY})$ defined by the given $p_{XY}$ and $p_{Z|XY}$ to one in normal form $(p_{X'Y'},p_{Z'|X'Y'})$ using the above described modification, with possibly smaller alphabets, so that any secure computation protocol for the former can be transformed to one for the latter with the same communication costs, and vice versa. The modification to the protocol, in either direction, is for each party to locally map $X$ to $X'$ etc., or vice versa; notice that the $Z'$ to $Z$ map is potentially randomized. Hence, it is enough to study secure computation problems $(p_{XY}, p_{Z|XY})$ in normal form.

\section{Entropy Lower Bounds as Communication Lower Bounds}\label{app:communication-lowerbound-through-entropy-lowerbound}
Our interest is in providing bounds on the amount of communication needed.
In this paper, we derive lower bounds on the entropies, $H(M_{12}),
H(M_{23}), H(M_{31})$, of the transcripts on the links. As we argue below,
these bounds will lowerbound the expected number of bits exchanged over the
link when we require that the message (bit string) sent by one party to
another in any round is a codeword in a prefix-free code such that the code
itself is determined by the messages exchanged between the two parties in
previous rounds. This is a natural requirement when variable length bit
strings are allowed as messages. Notice that without this restriction,
since a party sees the messages in each round rather than just a
concatenation of all the bits sent over the entire protocol, implicit
end-of-message markers are present which may convey implicit information
(without being accounted for in communication complexity which is measured
by the length of the transcript). The prefix-free requirement eliminates
the possibility of information being implicitly conveyed through end-of-message markers. 

Let us denote by $M_{\ij,t}$ the message sent {\em by} party-$i$ {\em to} party-$j$
in round-$t$. Let $L_{\ij,t}$ be the (potentially random) length in bits of
this message, and let $L_{ij} = \sum_{t=1}^N L_{\ij,t} + L_{\ji,t}$ be the
number of bits exchanged over the link $ij$ in either direction. We are
interested in lower bounds for ${\mathbb E}[L_{ij}]$. We have
\begin{align*}
H(M_{ij}) 
	&= \sum_{t=1}^\infty H(M_{\ij,t},M_{\ji,t}|M_{\ij}^{t-1},M_{\ji}^{t-1})\\
	&\leq \sum_{t=1}^\infty H(M_{\ij,t}|M_{\ij}^{t-1},M_{\ji}^{t-1}) +
H(M_{\ji,t}|M_{\ij}^{t-1},M_{\ji}^{t-1})\\
	&\stackrel{\text{(a)}}{\leq} \sum_{t=1}^\infty {\mathbb E}[L_{\ij,t}] +
{\mathbb E}[L_{\ji,t}]\\
	&= {\mathbb E}[L_{ij}],
\end{align*}
where (a) follows from the fact that the prefix-code of which $M_{\ij,t}$
is a codeword, is a function of the conditioning random variables
$(M_{\ij}^{t-1},M_{\ji}^{t-1})$, and hence the conditional entropy
$H(M_{\ij,t}|M_{\ij}^{t-1},M_{\ji}^{t-1})$ is no larger than ${\mathbb
E}[L_{\ij,t}]$ (by Kraft's inequality and non-negativity of
Kullback-Leibler divergence); similarly for the second term.

\section{Connections to Secure Sampling and Correlated Multi-Secret Sharing}\label{app:CMSS_sampling}

\paragraph{Secure Sampling.} In {\em secure sampling} functionalities, none of the parties receives any input, but all three parties produce outputs. The functionality is specified by a joint distribution $p_{XYZ}$ and the protocol for sampling $p_{XYZ}$ is specified by $\Pi(p_{XYZ})$.  The correctness condition in this case is that the outputs of Alice, Bob and Charlie are distributed according to $p_{XYZ}$. The security conditions remain the same as in the case of secure computation, that is, none of the parties can infer anything about the other parties' outputs other than what they can from their own outputs.

\paragraph{A Normal Form for $p_{XYZ}$.} For a joint distribution $p_{XYZ}$, define the relation $x\sim x'$ for $x,x'\in\X$ to hold if $\exists c\geq 0$ such that $\forall y\in\Y, z\in\Z$, $p(x,y,z)=c\cdot p(x',y,z)$. Similarly, we define $y\sim y'$ for $y,y'\in\Y$ and $z\sim z'$ for $z,z'\in\Z$. We say that $p_{XYZ}$ is in the {\em normal form} if $x\sim x' \Rightarrow x=x'$, $y\sim y' \Rightarrow y=y'$ and $z\sim z' \Rightarrow z=z'$.

It is easy to see that one can transform any distribution $p_{XYZ}$ to one in normal form $p_{X'Y'Z'}$, with possibly smaller alphabets, so that any secure sampling protocol for the former can be transformed to one for the latter with the same communication costs, and vice versa. (To define $X'$, $X$ is modified by removing all $x$ such that $p(x) = 0$ and then replacing all $x$ in an equivalence class of $\sim$ with a single representative; $Y'$ and $Z'$ are defined similarly. The modification to the protocol, in either direction, is for each party to locally map $X$ to $X'$ etc., or vice versa.) Hence it is enough to study the communication complexity of securely sampling distributions in the normal form.

Now, we show an analog of \Lemmaref{cutset} for secure sampling protocols.
\begin{lem}\label{lem:cutset_samp}
Suppose $p_{XYZ}$ is in normal form. Then, in any secure sampling protocol $\Pi(p_{XYZ})$,  the cut isolating Alice from Bob and Charlie must determine Alice's output $X$, i.e., $H(X| M_{12},M_{31}) =0$. Similarly, $H(Y|M_{12},M_{23})=0$ and $H(Z|M_{23},M_{31})=0$.
\begin{proof}
We only prove $H(X|M_{12},M_{31})=0$; the other ones, i.e., $H(Y|M_{12},M_{23})=0$ and $H(Z|M_{23},M_{31})=0$ are similarly proved. We need to show that for every $m_{12},m_{31}$
with $p(m_{12},m_{31})>0$, there is a (necessarily unique) $x\in {\mathcal
X}$ such that $p(x|m_{12},m_{31})=1$. Suppose, to the contrary, that we
have a secure sampling protocol resulting in a p.m.f. $p(x,y,z,m_{12},m_{31})$ such
that there exists $x,x'\in{\mathcal X}$, $x\neq x'$ and $m_{12},m_{31}$
satisfying $p(m_{12},m_{31})>0$, $p(x|m_{12},m_{31})>0$ and $p(x'|m_{12},m_{31})>0$. Since $p(m_{12},m_{31})>0$ and $p(x|m_{12},m_{31})>0$ imply $p_X(x)>0$, there exists $(y,z)$ s.t. $p_{XYZ}(x,y,z)>0$.
\begin{enumerate}
\item[(i)] The definition of a protocol implies that $p(x,y,z,m_{12},m_{31})$ can be written as \\ $p_{YZ}(y,z)p(m_{12},m_{31}|y,z)p(x|m_{12},m_{31})$.
\item[(ii)] Privacy against Alice implies that $p(x,y,z,m_{12},m_{31})$ can be written as \\ $p_{XYZ}(x,y,z)p(m_{12},m_{31}|x)$.
\item[(iii)] (i) and (ii) gives \\ $p_{YZ}(y,z)p(m_{12},m_{31}|y,z)p(x|m_{12},m_{31})=p_{XYZ}(x,y,z)p(m_{12},m_{31}|x)$. 
\end{enumerate}
By assumption, $p(m_{12},m_{31})>0$ and $p(x|m_{12},m_{31})>0$, which imply that $p(m_{12},m_{31}|x)>0$. And since $p_{XYZ}(x,y,z)>0$, we have from (iii) that $p(m_{12},m_{31}|y,z)>0$. Now consider $(x',y,z)$. By assumption, $p(m_{12},m_{31})>0$ and $p(x'|m_{12},m_{31})>0$, which imply $p(m_{12},m_{31}|x')>0$. Since $p(m_{12},m_{31}|y,z)>0$ from above, (iii) implies that $p_{XYZ}(x',y,z)>0$. Define $\alpha \triangleq \frac{p(x,y,z)}{p(x',y,z)}$. Since $p_{XYZ}$ is in normal form, $\exists (y',z')\in (\mathcal{Y,Z})$ s.t. $p_{XYZ}(x,y',z') \neq \alpha \cdot p_{XYZ}(x',y',z')$. Since $\alpha \neq 0$, at least one of $p(x,y',z')$ or $p(x',y',z')$ is non-zero. Assume that any one of these is non-zero, then applying the above arguments will give us that the other one should also be non-zero.
\begin{enumerate}
\item[(iv)] Dividing the expression in (iii) by the one we obtain when we apply the above arguments to $(x',y,z)$ gives $\frac{p(x|m_{12},m_{31})}{p(x'|m_{12},m_{31})} = \alpha \cdot \frac{p(m_{12},m_{31}|x)}{p(m_{12},m_{31}|x')}$.
\item[(v)] Repeating (i)-(iv) for $(x,y',z')$ and $(x',y',z')$, we get $\frac{p(x|m_{12},m_{31})}{p(x'|m_{12},m_{31})} \neq \alpha \cdot \frac{p(m_{12},m_{31}|x)}{p(m_{12},m_{31}|x')}$, which contradicts (iv).
\end{enumerate}
\end{proof}

\end{lem}
\begin{thm}\label{thm:samp_lbs}
Any secure sampling protocol $\Pi(p_{XYZ})$, where $p_{XYZ}$ is in normal form, should satisfy the following lower bounds on the entropy of the transcripts on each link.
\begin{align*}
H(M_{23}) &\geq RI(X;Z) + RI(X;Y) + H(Y,Z|X), \\
H(M_{31}) &\geq RI(Y;Z) + RI(X;Y) + H(X,Z|Y), \\
H(M_{12}) &\geq RI(X;Z) + RI(Y;Z) + H(X,Y|Z).
\end{align*}
\end{thm}
\begin{proof}
From \Lemmaref{cutset_samp}, we have $H(X|M_{12},M_{31})=0$, $H(Y|M_{12},M_{23})=0$ and $H(Z|M_{23},M_{31})=0$. Note that we can apply \Lemmaref{infoineq} for secure sampling of {\em dependent} $X$, $Y$ and $Z$, because, in the beginning, parties only have independent randomness, but no inputs. In the end, they output from a joint distribution $p_{XYZ}$, where $X$, $Y$ and $Z$ may be dependent, but this does not affect the requirements of \Lemmaref{infoineq} in any way. The proof for $H(M_{23})$ is given below; the other two bounds follows similarly.
\begin{align*}
H(M_{23}) &= I(M_{12};M_{23}) + H(M_{23}|M_{12}) \\
&= I(M_{12};M_{23}) + I(M_{31};M_{23}|M_{12}) + H(M_{23}|M_{12},M_{31})\\ 
&\stackrel{\text{(a)}}{\ge} I(M_{12};M_{23}|M_{31}) + I(M_{31};M_{23}|M_{12}) + H(M_{23}|M_{12},M_{31})\\ 
&\stackrel{\text{(b)}}{\ge} RI(X;Z) + RI(X;Y) + H(Y,Z | X),
\end{align*}
where (a) used $I(M_{12};M_{23}) \ge I(M_{12};M_{23}|M_{31})$, which follows from \Lemmaref{infoineq}; (b) used $I(M_{12};M_{23}|M_{31}) \geq RI(X;Z)$, $I(M_{31};M_{23}|M_{12}) \geq RI(X;Y)$ and $H(M_{23}|M_{12},M_{31}) \geq H(Y,Z | X)$, which we have shown in the proof of \Theoremref{prelim_lbs}.
\end{proof}

We remark that if the marginal distributions satisfy $p_{XY} = p_X p_Y$ (i.e., $X$ and $Y$ are independent), then a secure computation protocol for $p_{Z|XY}$ can be turned into a secure sampling protocol (with the same communication costs), by having Alice and Bob locally sample inputs $X$ and $Y$ according to $p_X$ and $p_Y$ and then run the computation protocol. So, whenever $X$ and $Y$ are independent, the lower bounds on communication for secure sampling imply lower bounds for secure computation.

\paragraph{Correlated Multi-Secret Sharing Schemes.}
We define a notion of secret-sharing, called Correlated Multi-Secret Sharing (CMSS) that is closely related to secure sampling/computation problem. We will show that lower bounds on the entropy of shares of such secret-sharing schemes will also be lower bounds on entropy of transcripts for the corresponding secure computation protocols. However, we shall show a separation between the efficiency of secret-sharing (where there is an omniscient dealer) and a protocol, using the stronger lower bounds we have established in \Sectionref{lowerbounds}.
\begin{defn}
Given a graph $G=(V,E)$, an adversary structure $\mathcal{A} \subseteq 2^V$, and a joint distribution $p_{(X_v)_{v\in V}}$ over random variables $X_v$ indexed by $v\in V$, a correlated multiple secret sharing scheme for $(G,p_{(X_v)_{v\in V}})$ defines a distribution $p_{(M_e)_{e\in E}|(X_v)_{v\in V}}$ of shares $M_e$ for each edge $e\in E$, such that the following hold. Below, for $S\subseteq E$, $M_S$ stands for the collection of all $M_e$ for $e\in
S$; similarly $X_T$ is defined for $T\subseteq V$; $E_v\subseteq E$ denotes the set of edges incident on a vertex $V$.
\begin{itemize}
\item Correctness: For all $v\in V$, $H(X_v|M_{E_v})=0$.
\item Privacy: For every set $T\in \mathcal{A}$, let $E_T=\cup_{v\in T} E_v$; then, $I(X_{\overline{T}};M_{E_T}|X_T)= 0$.
\end{itemize}
\end{defn}
Below we give a specialised version of the above general definition which is suitable to our setting, where $G$ is the clique over the vertex set $V=\{1,2,3\}$, and $\mathcal{A}=\{ \{1\}, \{2\}, \{3\} \}$ (corresponding to 1-privacy).

We define $\Sigma$ to be a {\em correlated multi-secret sharing scheme for a joint distribution $p_{XYZ}$} (with respect to our fixed adversary structures) if it probabilistically maps secrets $(X,Y,Z)$ to shares $M_{12},M_{23},M_{31}$ such that the following conditions hold:
\begin{itemize}
\item Correctness: $H(X|M_{12},M_{31}) = H(Y|M_{12},M_{23}) = H(Z|M_{23},M_{31}) = 0$.
\item Privacy: 
\begin{align*}
I((M_{12},M_{31});(Y,Z)|X) &=0 \qquad \text{(privacy against Alice)}, \\
I((M_{12},M_{23});(X,Z)|Y) &=0 \qquad \text{(privacy against Bob)}, \\
I((M_{23},M_{31});(X,Y)|Z) &=0 \qquad \text{(privacy against Charlie)}.
\end{align*}

\end{itemize}
We point out that while the correctness condition relates only to the supports of $X$, $Y$ and $Z$ individually, the privacy condition is crucially influenced by the joint distribution.

\begin{thm}\label{thm:CMSS_prelim_lbs}
Any CMSS scheme for any joint distribution $p_{XYZ}$ satisfies
\begin{align*}
H(M_{12}) &\geq \max\{RI(X;Z), RI(Y;Z)\} + H(X,Y|Z), \\
H(M_{23}) &\geq \max\{RI(X;Z), RI(X;Y)\} + H(Y,Z|X), \\
H(M_{31}) &\geq \max\{RI(Y;Z), RI(X;Y)\} + H(X,Z|Y).
\end{align*}
\end{thm}
\begin{proof}
We proceed along the lines of the proof of \Theoremref{prelim_lbs}, except that here we do not need \Lemmaref{cutset} to argue that $H(X|M_{12},M_{31})$ = $H(Y|M_{12},M_{23})$ = $H(Z|M_{23},M_{31}) = 0$, instead, these follow from the correctness of CMSS.
\end{proof}

If $p_{XYZ}=p_{XY}p_{Z|XY}$, where $(p_{XY},p_{Z|XY})$ is in normal form, using \Lemmaref{cutset}, the bounds in \Theoremref{CMSS_prelim_lbs} imply bounds in \Theoremref{prelim_lbs}. If $p_{XYZ}$ has full support, then we can further strengthen the bounds in \Theoremref{CMSS_prelim_lbs} by applying distribution switching.

\begin{thm}\label{thm:CMSS_lbs}
Consider any CMSS scheme for a joint distribution $p_{XYZ}$, where $p_{XYZ}$ has full support.
\begin{enumerate}
\item $H(M_{12}) \geq \max\{\sup_{p_{X'Y'Z'}}\left(RI(X';Z')+H(X',Y'|Z')\right), \sup_{p_{X'Y'Z'}}\left(RI(Y';Z') + H(X',Y'|Z')\right)\},$ \\[0.2cm]
where $p_{X'Y'Z'}$ is any distribution for which the characteristic bipartite graph of $p_{X'Y'}$ is connected.
\item $H(M_{23}) \geq \max\{\sup_{p_{X'Y'Z'}}\left(RI(X';Z')+H(Y',Z'|X')\right), \sup_{p_{X'Y'Z'}}\left(RI(X';Y') + H(Y',Z'|X')\right)\},$ \\[0.2cm]
where $p_{X'Y'Z'}$ is any distribution for which the characteristic bipartite graph of $p_{Y'Z'}$ is connected.
\item $H(M_{31}) \geq \max\{\sup_{p_{X'Y'Z'}}\left(RI(Y';Z')+H(X',Z'|Y')\right), \sup_{p_{X'Y'Z'}}\left(RI(X';Y') + H(X',Z'|Y')\right)\},$ \\[0.2cm]
where $p_{X'Y'Z'}$ is any distribution for which the characteristic bipartite graph of $p_{X'Z'}$ is connected.
\end{enumerate}
\end{thm}
\begin{proof}
First we observe that we can apply distribution switching to CMSS schemes also, i.e., if we have a CMSS $\Sigma(p_{XYZ})$, where $p_{XYZ}$ has full support, it will remain a CMSS if we change the distribution to a different one $p_{X'Y'Z'}$. This follows from the correctness and privacy conditions of a CMSS. Proceeding as in the proof of \Lemmaref{XYZ_inde_M123}, we can show that for any CMSS $\Sigma(p_{XYZ})$, connectedness of the characteristic bipartite graph of $p_{XY}$ implies $I(X,Y,Z;M_{12})=0$. The other two, i.e., connectedness of the characteristic bipartite graph of $p_{XZ}$ implies $I(X,Y,Z;M_{31})=0$ and connectedness of the characteristic bipartite graph of $p_{YZ}$ implies $I(X,Y,Z;M_{23})=0$, follow similarly. Now, we can apply the distribution switching to the bounds in \Theoremref{CMSS_prelim_lbs}.
\end{proof}

It is easy to see that any secure sampling protocol $\Pi(p_{XYZ})$, where $p_{XYZ}$ is in normal form, yields a CMSS scheme for
the same joint distribution $p_{XYZ}$: An omniscient dealer can always
produce the shares $M_{12},M_{23},M_{31}$ which are precisely the
transcripts produced by the secure sampling protocol. Now, correctness for
this CMSS follows from \Lemmaref{cutset_samp}, and
privacy of CMSS scheme follows from the privacy of the secure sampling
protocol. Thus the lower bounds on the transcripts produced by a CMSS
scheme for a given $p_{XYZ}$ in normal form, gives lower bounds on the corresponding links for any secure sampling protocol for this $p_{XYZ}$. If $p_{XYZ}=p_{XY}p_{Y|XY}$, where $(p_{XY},p_{Z|XY})$ is in normal form, then lower bounds for CMSS schemes provide lower bounds for secure
computation problems. As we discuss in \Pageref{thm:and}, this lower
bound is not tight in general, i.e., there is a function (in fact the {\sc
and} function) for which there is a CMSS scheme which requires less
communication than what our lower bound for secure computation for that
function provides. Towards this, here we give upper bounds on the share
sizes of a 3-party CMSS for {\sc and} which is defined as $X$ and $Y$
independent and uniformly distributed bits, and $Z=X\land Y$.

\begin{thm}\label{thm:gap_CMSS}
For $p_{XYZ}$ such that $X$ and $Y$ independent and uniformly distributed
bits, and $Z=X\land Y$, there is a CMSS $\Sigma(p_{XYZ})$ which has
$H(M_{12})=H(M_{23})=H(M_{31})=\log(3)$.
\end{thm}
\begin{proof}

Consider a CMSS scheme $\Sigma$ defined as follows.
Let $(\alpha,\beta,\gamma)$ be a random permutation of the set
$\{0,1,2\}$. Let $M_{12}=\alpha$, and
\begin{align*}
M_{31}=\begin{cases}
\alpha \qquad \text{if }X=1 \\
\beta  \qquad \text{if }X=0
\end{cases}
&
\qquad
M_{23}=\begin{cases}
\alpha \qquad \text{if }Y=1 \\
\gamma  \qquad \text{if }Y=0
\end{cases}
\end{align*}
It can be seen that this scheme satisfies the correctness and privacy
requirements (in particular, 
$(M_{12},M_{31})$ is uniformly random, conditioned on
$M_{12}=M_{31}$ when $X=1$ and conditioned on
$M_{12}\not=M_{31}$  when $X=0$).
$H(M^\Sigma_{12})=H(M^\Sigma_{23})=H(M^\Sigma_{31})=\log 3 < 1.585$.\\

\noindent \Theoremref{CMSS_lbs} implies that this scheme is optimal.
\end{proof}

\section{Proofs of the Main Theorems}\label{app:main_thms_proofs}
\begin{proof}[Proof of \Theoremref{lowerbound}]
Here, we prove our lower bounds only for {\em independent} inputs, i.e., $p_{XY}=p_Xp_Y$, but as we show in \Appendixref{dependent_inputs_lbs}, they also hold for dependent inputs $p_{XY}$ with full support.

Suppose we have a secure protocol for computing $p_{Z|XY}$ in the normal form under $p_X,p_Y$ which have full support. Consider $H(M_{23})$,
\[ H(M_{23}) = I(M_{23};M_{12}) + I(M_{23};M_{31}|M_{12}) + H(M_{23}|M_{12},M_{31}).\]
By \Lemmaref{XY_inde_M23M31}, $M_{23}$ is independent of $X$. So, by distribution switching, we know that we may switch the distribution of $X$ to, say, $p_{X''}$ which also has full support and the resulting $M_{23}$ has the same distribution as under $p_X$, i.e.,
\begin{align*}
H(M_{23}) = \sup_{p_{X''}} I(M_{23};M_{12}) + I(M_{23};M_{31}|M_{12}) + H(M_{23}|M_{12},M_{31}).
\end{align*}
Under this switched distribution, let us consider the first term $I(M_{23};M_{12})$. Let us notice that, by privacy against Bob, $(M_{23},M_{12})$ must again be independent of $X''$. Hence, even if we switch the distribution of $X$ to, say $p_{X'}$, the joint distribution of $(M_{23},M_{12})$ must remain unchanged. Hence, we have that $I(M_{23};M_{12})$ under the distribution $p_{X''}$ is the same as that under $p_{X'}$. Therefore,
\begin{align*}
H(M_{23}) = \left(\sup_{p_{X'}} I(M_{23};M_{12})\right) + \left(\sup_{p_{X''}} I(M_{23};M_{31}|M_{12}) +
H(M_{23}|M_{12},M_{31})\right).
\end{align*}
Now proceeding as in the proof of \Theoremref{intermediate_lbs}, we have
\begin{align*}
H(M_{23}) \geq \left(\sup_{p_{X'}} RI(X';Z')\right) + \left(\sup_{p_{X''}} H(Y,Z''|X'')\right).
\end{align*}
The bound on $H(M_{31})$ follows in an identical fashion. To see the bounds on $H(M_{12})$, let us recall that $M_{12}$ is independent of $X,Y$ (by \Lemmaref{XYZ_inde_M123}) and hence we may switch the distributions of both $X$ and $Y$. Furthermore, let us note that we may write $H(M_{12})$ in two different ways.
\begin{align}
H(M_{12}) &= [I(M_{12};M_{31})] +
[I(M_{12};M_{23}|M_{31})+H(M_{12}|M_{23},M_{31})]\label{eq:M12intwoways1}\\
H(M_{12}) &= [I(M_{12};M_{23})] +
[I(M_{12};M_{31}|M_{23})+H(M_{12}|M_{23},M_{31})].\label{eq:M12intwoways2}
\end{align}
Using \eqref{eq:M12intwoways1} and proceeding as we did for $H(M_{23})$ leads to the top row of the right hand side of \eqref{eq:M12_bound}, and \eqref{eq:M12intwoways2} leads to the bottom row. 
\end{proof}

\begin{proof}[Proof of \Theoremref{conditionallowerbounds}]
Here, we prove our lower bounds only for {\em independent} inputs, i.e., $p_{XY}=p_Xp_Y$, but as we show in \Appendixref{dependent_inputs_lbs}, they also hold for dependent inputs $p_{XY}$ with full support.

In the proof of \Lemmaref{XYZ_inde_M123}, we show that under condition~1, $M_{31}$ is independent of both $X,Y$. This allows us to switch the distribution of both $X$ and $Y$ as we did for bounding $H(M_{12})$ in the proof of \Theoremref{lowerbound}. Proceeding in an identical fashion as there leads us to~\eqref{eq:M31_WithCond}. Similarly, under condition~2, $M_{23}$ is independent of $X,Y$ which leads to~\eqref{eq:M23_WithCond}.
\end{proof}

\section{Proofs Omitted from \Sectionref{lowerbounds}}\label{app:proofs}
\begin{proof}[Proof of \Lemmaref{cutset}]
First we will show $H(X|M_{12},M_{13})=0$; the other one, i.e. $H(Y|M_{12},M_{23})=0$, is similarly proved. We apply a cut-set argument. Consider the cut isolating Alice from Bob and Charlie.

We need to show that for every $m_{12},m_{31}$ with $p(m_{12},m_{31})>0$,
there is a (necessarily unique) $x\in {\mathcal X}$ such that
$p(x|m_{12},m_{31})=1$. Suppose, to the contrary, that we have a secure
protocol resulting in a p.m.f. $p(x,y,z,m_{12},m_{31})$ such that there
exists $x,x'\in{\mathcal X}$, $x\neq x'$ and $m_{12},m_{31}$ satisfying
$p(m_{12},m_{31})>0$, $p(x|m_{12},m_{31})>0$ and $p(x'|m_{12},m_{31})>0$. For these $x,x'$, since
$(p_{XY},p_{Z|XY})$ is in the normal form, $\exists (y,z) \in
{\mathcal{Y}\times\mathcal{Z}}$ such that $p_{XY}(x,y)>0,p_{XY}(x',y)>0$ and $p_{Z|X,Y}(z|x,y) \neq
p_{Z|X,Y}(z|x',y)$.

\begin{enumerate}
\item[(i)] The definition of a protocol implies that $p(x,y,z,m_{12},m_{31})$ can be written as \\ $p_{X,Y}(x,y)p(m_{12},m_{31}|x,y)p(z|m_{12},m_{31},y)$.
\item[(ii)] Privacy against Alice implies that $p(m_{12},m_{31}|x,y,z)=p(m_{12},m_{31}|x)$.
\item[(iii)] Using (ii) in (i), we get $p(x,y,z,m_{12},m_{31})=p_{X,Y}(x,y)p(m_{12},m_{31}|x)p(z|m_{12},m_{31},y)$.
\item[(iv)] Correctness and (ii) imply that we can also write\\$p(x,y,z,m_{12},m_{31})=p_{X,Y}(x,y)p_{Z|X,Y}(z|x,y)p(m_{12},m_{31}|x)$.
\item[(v)] Since $p_{X,Y}(x,y)p(m_{12},m_{31}|x)>0$, from (iii) and (iv), we get $p(z|m_{12},m_{31},y)=p_{Z|X,Y}(z|x,y)$.
\end{enumerate}
Applying the above arguments to $(x',y,z,m_{12},m_{31})$, we get $p(z|m_{12},m_{31},y)=p_{Z|X,Y}(z|x',y)$, leading to the contradiction $p(z|m_{12},m_{31},y) \neq p(z|m_{12},m_{31},y)$, since by assumption $p_{Z|X,Y}(z|x,y) \neq p_{Z|X,Y}(z|x',y)$.

For $H(Z|M_{23},M_{31})=0$, we need to show that for every $m_{23},m_{31}$
with $p(m_{23},m_{31})>0$, there is a (necessarily unique) $z\in {\mathcal
Z}$ such that $p(z|m_{23},m_{31})=1$. Suppose, to the contrary, that we
have a secure protocol resulting in a p.m.f. $p(x,y,z,m_{23},m_{31})$ such
that there exists $z,z'\in{\mathcal Z}$, $z\neq z'$ and $m_{23},m_{31}$
satisfying $p(m_{23},m_{31})>0$, $p(z|m_{23},m_{31})$ and $p(z'|m_{23},m_{31})>0$. Since $(p_{XY},p_{Z|XY})$ is in normal form, there exists $(x,y)$ s.t. $p_{XY}(x,y)>0$ and $p_{Z|X,Y}(z|x,y)>0$.
\begin{enumerate}
\item[(i)] The definition of a protocol implies that $p(x,y,z,m_{23},m_{31})$ can be written as \\ $p_{X,Y}(x,y)p(m_{23},m_{31}|x,y)p(z|m_{23},m_{31})$.
\item[(ii)] Privacy against Charlie implies that $p(x,y,z,m_{23},m_{31})$ can be written as \\ $p_{X,Y}(x,y)p(z|x,y)p(m_{23},m_{31}|z)$.
\item[(iii)] (i) and (ii) gives \\ $p(m_{23},m_{31}|x,y)p(z|m_{23},m_{31})=p_{Z|X,Y}(z|x,y)p(m_{23},m_{31}|z)$. 
\end{enumerate}
By assumption, $p(m_{23},m_{31})>0$ and $p(z|m_{23},m_{31})>0$, which imply that $p(m_{23},m_{31}|z)>0$. And since $p_{Z|X,Y}(z|x,y)>0$, we have from (iii) that $p(m_{23},m_{31}|x,y)>0$. Now consider $(x,y,z')$. By assumption, $p(m_{23},m_{31})>0$ and $p(z'|m_{23},m_{31})>0$, which imply $p(m_{23},m_{31}|z')>0$. Since $p(m_{23},m_{31}|x,y)>0$ from above, (iii) implies that $p_{Z|X,Y}(z'|x,y)>0$. Define $\alpha \triangleq \frac{p(z|x,y)}{p(z'|x,y)}$. Since $(p_{XY},p_{Z|XY})$ is in normal form, $\exists (x',y')\in (\mathcal{X,Y})$ s.t. $p_{XY}(x',y')>0$ and $p_{Z|X,Y}(z|x',y') \neq \alpha \cdot p_{Z|X,Y}(z'|x',y')$. Since $\alpha \neq 0$, at least one of $p(z|x',y')$ or $p(z'|x',y')$ is non-zero. Assume that any one of these is non-zero, then applying the above arguments will give us that the other one should also be non-zero.
\begin{enumerate}
\item[(iv)] Dividing the expression in (iii) by the one we obtain when we apply the above arguments to $(x,y,z')$ gives $\frac{p(z|m_{23},m_{31})}{p(z'|m_{23},m_{31})} = \alpha \cdot \frac{p(m_{23},m_{31}|z)}{p(m_{23},m_{31}|z')}$.
\item[(v)] Repeating (i)-(iv) for $(x',y',z)$ and $(x',y',z')$, we get $\frac{p(z|m_{23},m_{31})}{p(z'|m_{23},m_{31})} \neq \alpha \cdot \frac{p(m_{23},m_{31}|z)}{p(m_{23},m_{31}|z')}$, which contradicts (iv).
\end{enumerate}
\end{proof}

\begin{proof}[Proof of \Lemmaref{XYZ_inde_M123}]\ 
\begin{enumerate}
\item To show $I(X,Y,Z;M_{12})=0$, we need only show that $I(X;M_{12})=0$, since $I(X,Y,Z;M_{12})=I(X;M_{12})+I(Y,Z;M_{12}|X)$ and the second term is equal to zero by the privacy against Alice.

For $I(X;M_{12})=0$, we need to show that $p(m_{12}|x)=p(m_{12}|x')$ for all $x,x'\in{\mathcal
X}$. Take some $x,x' \in \mathcal{X}, x \neq x'$. Suppose there is a $y \in
\mathcal{Y}$ s.t. $p_{XY}(x,y)>0,p_{XY}(x',y) >0$. Then, by privacy against Alice
$p(m_{12},x,y)=p_{X,Y}(x,y)p(m_{12}|x)$ and by privacy against
Bob $p(m_{12},x,y)= p_{X,Y}(x,y)p(m_{12}|y)$. By comparing these two,
we get $p(m_{12}|x)=p(m_{12}|y)$. Applying the above arguments to
$(m_{12},x',y)$ gives $p(m_{12}|x') = p(m_{12}|y)$. Hence,
$p(m_{12}|x)=p(m_{12}|x')$.

Connectedness of the characteristic bipartite graph of $p_{XY}$ implies that for every $x,x' \in \mathcal{X}$, there is a sequence $x_0=x,x_1,x_2,\hdots,x_{L-1},x_L=x'\in \mathcal{X}$ such that for every pair $(x_{l-1},x_l), l=1,2,\hdots,L$, there is a $y_l\in \mathcal{Y}$ s.t. $p_{X,Y}(x_{l-1},y_l)>0$ and $p_{X,Y}(x_l,y_l)>0$. Hence, $p(m_{12}|x)=p(m_{12}|x_1)=$ $p(m_{12}|x_2)=\hdots =p(m_{12}|x')$.\\

\item To show $I(X,Y,Z;M_{31})=0$ under condition~1, we need only show that $I(X;M_{31})=0$, since $I(X,Y,Z;M_{31})=I(X;M_{31})+I(Y,Z;M_{31}|X)$ and the second term is equal to zero by the privacy against Alice.

We need to show that $p(m_{31}|x)=p(m_{31}|x')$ for all $x,x'\in{\mathcal
X}$. Take some $x,x' \in \mathcal{X}, x \neq x'$. Suppose there is a $z \in
\mathcal{Z}$ s.t. $p_{Z|X,Y}(z|x,y),p_{Z|X,Y}(z|x',y') >0$ for some $y,y'
\in \mathcal{Y}$. Then, by privacy against Alice
$p(m_{31},x,z)=p_{X,Z}(x,z)p(m_{31}|x)$ and by privacy against
Charlie $p(m_{31},x,z)= p_{X,Z}(x,z)p(m_{31}|z)$. By comparing these two,
and since $p_{X,Z}(x,z) >0$ (which follows from the assumption that $p_{X,Y}$ has full
support), we get $p(m_{31}|x)=p(m_{31}|z)$. Applying the above arguments to
$(x',z)$ gives $p(m_{31}|x') = p(m_{31}|z)$. Hence,
$p(m_{31}|x)=p(m_{31}|x')$.

Condition 1 implies that for every $x,x' \in \mathcal{X}$, there is a sequence $x_0=x,x_1,x_2,\hdots,x_{L-1},x_L=x'\in \mathcal{X}$ such that for every pair $(x_{l-1},x_l), l=1,2,\hdots,L$, there is a $z_l\in \mathcal{Z}$ s.t. $p_{Z|X,Y}(z_l | x_{l-1},y_l)$, $p_{Z|X,Y}(z_l |x_l,y_l') >0$ for some $y_l,y_l' \in \mathcal{Y}$. Hence, $p(m_{31}|x)=p(m_{31}|x_1)=$ $p(m_{31}|x_2)=\hdots =p(m_{31}|x')$.

\item The other case under condition~2 follows similarly.
\end{enumerate}
\end{proof}

\begin{proof}[Proof of \Lemmaref{infoineq}] We will apply induction on the number of rounds of the protocol.\\
{\em Base case:} At the beginning of the protocol, all the transcripts $M_{\gamma \alpha}, M_{\beta \gamma}$ and $M_{\alpha \beta}$ are empty. So, the inequality is trivially true. \\
{\em Inductive step:} Assume that the inequality is true at the end of round $t$, and we prove it for $t+1$. For simplicity, let us denote the transcript $M_{\gamma \alpha}$ (similarly others) at the end of round $t$ by $M_{\gamma \alpha}$ itself and at the end of round $t+1$ by $\widetilde{M}_{\gamma \alpha}$. We denote by $\Delta M$, the new message sent in round $t+1$ and if that message is sent from party $\gamma$ to party $\alpha$, we denote it by $\Delta M_{\vec{\gamma \alpha}}$ and so $\widetilde{M}_{\gamma \alpha}$ becomes $(M_{\gamma \alpha}, \Delta M_{\vec{\gamma \alpha}})$.\\

Observe that we need to consider only three kinds of messages exchanged in round $t+1$, which are $\Delta M_{\vec{\beta \alpha}}, \Delta M_{\vec{\beta \gamma}}$ and $\Delta M_{\vec{\gamma \beta}}$. The inequality for other three kinds of messages is similarly proved. Since the parties do not share any common or correlated randomness, the new message that one party (say, $\beta$) sends to another (say, $\alpha$) is conditionally independent of the transcript ($M_{\gamma\alpha}$) between the other two parties ($\gamma$ and $\alpha$) conditioned on the transcripts ($M_{\alpha\beta},M_{\beta\gamma}$) on both of the links to which that party (namely, $\beta$) is associated with. So we have the following:
\begin{align}
I(M_{\gamma \alpha}; \Delta M_{\vec{\beta \alpha}}| M_{\alpha \beta}, M_{\beta \gamma}) &= 0, \label{eq:21inde} \\
I(M_{\gamma \alpha}; \Delta M_{\vec{\beta \gamma}}| M_{\alpha \beta}, M_{\beta \gamma}) &= 0, \label{eq:23inde} \\
I(M_{\alpha \beta}; \Delta M_{\vec{\gamma \beta}}| M_{\beta \gamma}, M_{\gamma \alpha}) &= 0. \label{eq:32inde}
\end{align}

\begin{enumerate}
\item [1.] If $\Delta M = \Delta M_{\vec{\beta \alpha}}$, then
\begin{align*}
I(\widetilde{M}_{\gamma \alpha}; \widetilde{M}_{\beta \gamma}) &\stackrel{\text{(a)}}{=} I({M}_{\gamma \alpha}; {M}_{\beta \gamma}) \\
&\stackrel{\text{(b)}}{\geq} I(M_{\gamma \alpha}; M_{\beta \gamma} | M_{\alpha \beta}) \\
&\stackrel{(\text{c)}}{=} I(M_{\gamma \alpha}; M_{\beta \gamma}, \Delta M_{\vec{\beta \alpha}} | M_{\alpha \beta}) \\
&\geq I(\underbrace{M_{\gamma \alpha}}_{\widetilde{M}_{\gamma \alpha}}; \underbrace{M_{\beta \gamma}}_{\widetilde{M}_{\beta \gamma}} | \underbrace{M_{\alpha \beta}, \Delta M_{\vec{\beta \alpha}}}_{\widetilde{M}_{\alpha \beta}}) \\
&= I(\widetilde{M}_{\gamma \alpha}; {\widetilde{M}_{\beta \gamma}} |
\widetilde{M}_{\alpha \beta}),
\end{align*}
where (a) follows because $\widetilde{M}_{\gamma \alpha} = M_{\gamma \alpha}$ and $\widetilde{M}_{\beta \gamma} = M_{\beta \gamma}$, (b) follows from the induction hypothesis and ({c}) follows from \eqref{eq:21inde}.\\

\item [2.] If $\Delta M = \Delta M_{\vec{\beta \gamma}}$, then
\begin{align*}
I(\widetilde{M}_{\gamma \alpha}; \widetilde{M}_{\beta \gamma}) &\stackrel{\text{(d)}}{=} I({M}_{\gamma \alpha}; {M}_{\beta \gamma}, \Delta M_{\vec{\beta \gamma}}) \\
&\geq I({M}_{\gamma \alpha}; {M}_{\beta \gamma}) \\
&\stackrel{\text{(e)}}{\geq} I(M_{\gamma \alpha}; M_{\beta \gamma} | M_{\alpha \beta}) \\
&\stackrel{\text{(f)}}{\geq} I(\underbrace{M_{\gamma \alpha}}_{\widetilde{M}_{\gamma \alpha}}; \underbrace{M_{\beta \gamma},\Delta M_{\vec{\beta \gamma}}}_{\widetilde{M}_{\beta \gamma}} | \underbrace{M_{\alpha \beta}}_{\widetilde{M}_{\alpha \beta}}) \\
&= I(\widetilde{M}_{\gamma \alpha}; {\widetilde{M}_{\beta \gamma}} |
\widetilde{M}_{\alpha \beta}),
\end{align*}
where (d) follows because $\widetilde{M}_{\gamma \alpha} = M_{\gamma \alpha}$ and $\widetilde{M}_{\beta \gamma} = (M_{\beta \gamma}, \Delta M_{\vec{\beta \gamma}})$, (e) follows from the induction hypothesis and (f) follows from \eqref{eq:23inde}.\\

\item [3.] If $\Delta M = \Delta M_{\vec{\gamma \beta}}$, then
\begin{align*}
I(\widetilde{M}_{\gamma \alpha}; \widetilde{M}_{\beta \gamma}) &\stackrel{\text{(g)}}{=} I({M}_{\gamma \alpha}; {M}_{\beta \gamma}, \Delta M_{\vec{\gamma \beta}}) \\
&= I({M}_{\gamma \alpha}; {M}_{\beta \gamma}) + I({M}_{\gamma \alpha}; \Delta M_{\vec{\gamma \beta}} | {M}_{\beta \gamma}) \\
&\stackrel{\text{(h)}}{=} I({M}_{\gamma \alpha}; {M}_{\beta \gamma}) + I({M}_{\gamma \alpha}; \Delta M_{\vec{\gamma \beta}} | {M}_{\beta \gamma}) + I(M_{\alpha \beta}; \Delta M_{\vec{\gamma \beta}}| M_{\beta \gamma}, M_{\gamma \alpha}) \\
&= I({M}_{\gamma \alpha}; {M}_{\beta \gamma}) + I({M}_{\gamma \alpha}, M_{\alpha \beta}; \Delta M_{\vec{\gamma \beta}} | {M}_{\beta \gamma}) \\
&\stackrel{\text{(i)}}{\geq} I(M_{\gamma \alpha}; M_{\beta \gamma} | M_{\alpha \beta}) + I({M}_{\gamma \alpha}; \Delta M_{\vec{\gamma \beta}} | {M}_{\beta \gamma}, M_{\alpha \beta}) \\
&= I(\underbrace{M_{\gamma \alpha}}_{\widetilde{M}_{\gamma \alpha}}; \underbrace{M_{\beta \gamma}, \Delta M_{\vec{\gamma \beta}}}_{\widetilde{M}_{\beta \gamma}} | \underbrace{M_{\alpha \beta}}_{\widetilde{M}_{\alpha \beta}}) \\
&= I(\widetilde{M}_{\gamma \alpha}; \widetilde{M}_{\beta \gamma} |
\widetilde{M}_{\alpha \beta}),
\end{align*}
where (g) follows because $\widetilde{M}_{\gamma \alpha} = M_{\gamma \alpha}$ and $\widetilde{M}_{\beta \gamma} = (M_{\beta \gamma}, \Delta M_{\vec{\gamma \beta}})$, (h) follows from \eqref{eq:32inde} and (i) follows from the induction hypothesis.
\end{enumerate}

\end{proof}

\section{Details omitted from \Sectionref{examples}}\label{app:examples}

\renewcommand{\star}{+}
\subsection{Secure Computation of {\sc group-add}} \label{app:group}
Let $\mathbb{G}$ be a (possibly non-abelian) group with binary operation $\star$. The 
function {\sc group-add} is defined as follows: Alice has an input $X\in\mathbb{G}$, Bob has an input
$Y\in\mathbb{G}$ and Charlie should get $Z=f(X,Y)=X \star Y$.

In \Figureref{group}, we recapitulate a well-known simple protocol for
securely computing the above function. The protocol requires a
$|\mathbb{G}|$-ary symbol to be exchanged per computation over each link. As
we show below, this protocol is easily seen to be optimal in terms of
expected number of bits on each link as well as the amount of randomness. For vectors $X,Y\in\mathbb{G}^n$, we write $X\star Y$ to denote the component-wise computation.

\begin{figure}[htb]
\hrule height 1pt
\vspace{.05cm}
{\bf Algorithm 1:} {Secure Computation of {\sc group-add}} 
\hrule
\begin{algorithmic}[1]
\REQUIRE Alice \& Bob have input vectors $X,Y\in\mathbb{G}^n$.
\ENSURE Charlie securely computes the component-wise \[Z=X\star Y.\]

\medskip

\STATE Charlie samples $n$ i.i.d. uniformly distributed elements $K=(K_1,K_2,\hdots,K_n)$ from $\mathbb{G}$ using his private randomness; sends it to Bob as $M_{\vec{32}}=K$.

\STATE Bob sends $M_{\vec{21}}=Y \star M_{\vec{32}}$ to Alice.

\STATE Alice sends $M_{\vec{13}}=X \star M_{\vec{21}}$ to Charlie.

\STATE Charlie outputs $Z= M_{\vec{13}} - K$.
\end{algorithmic}
\hrule
\caption{An optimal protocol for secure computation in any group
$\mathbb{G}$. The protocol requires a $|\mathbb{G}|$-ary symbol to be exchanged per computation over each link.}
\label{fig:group}
\end{figure}
\begin{thm}\label{thm:group}
Any secure protocol for computing in a Group $\mathbb{G}$, where $p_{XY}$ has full support over $\mathbb{G}^n\times\mathbb{G}^n$, must satisfy 
\begin{align*}
H(M_{12}), H(M_{23}), H(M_{31}) \geq n\log|\mathbb{G}|, \\
\rho(\text{\sc GROUP-ADD}) \geq n\log|\mathbb{G}|.
\end{align*}
\end{thm}
\begin{proof}
It is easy to see that the above function satisfies Condition~1 and Condition~2 of \Lemmaref{XYZ_inde_M123}. 
We will only need the last terms (corresponding to the na\"ive bounds $H(X',Y'|Z')$ etc., but
with distribution switching) of \eqref{eq:improved_prelim_lb_M12}, \eqref{eq:improved_prelim_lb_M31} and \eqref{eq:improved_prelim_lb_M23} for $H(M_{12})$, $H(M_{31})$ and $H(M_{23})$ respectively.
Since we are computing a deterministic function, and $Y$ can be determined from $(X,Z)$, the last terms in each of the these bounds will reduce to the following:
\begin{align*}
H(M_{12}) &\geq \sup_{p_{X'Y'}}  H(X'|Z'),\\
H(M_{31}) &\geq \sup_{p_{X'Y'}}  H(X'|Y'),\\
H(M_{23}) &\geq \sup_{p_{X'Y'}}  H(Y'|X').
\end{align*}
The optimum bounds for $M_{12}$, $M_{31}$ and $M_{23}$ are obtained by taking $X'$ and $Y'$ to be independent and uniform over $\mathbb{G}^n$, which gives $H(M_{12}), H(M_{31}), H(M_{23}) \geq n\log|\mathbb{G}|$.

From~\Theoremref{randomness} and the above bound on $H(M_{12})$, we have $\rho(\text{\sc GROUP-ADD}) \geq n\log |\mathbb{G}|$, which implies that the above protocol is randomness-optimal.
\end{proof}


\subsection{Secure Computation of {\sc sum}}\label{app:sum}
The {\sc sum} function is defined as follows: Alice and Bob have one bit input $X\in\{0,1\}$ and $Y\in\{0,1\}$ respectively. Charlie wants to compute the arithmetic sum $Z=f(X,Y)=X+Y$. \Figureref{sum} recapitulates a simple protocol for this function. This protocol requires a ternary symbol to be exchanged per computation over each link. We show in below that our bounds give $H(M_{31}),H(M_{23})\geq \log(3)$ and $H(M_{12})\geq 1.5$. Thus, while the protocol matches the lower bound on $H(M_{31})$ and $H(M_{23})$, there is a gap for $H(M_{12})$. While the protocol requires $H(M_{12})=\log(3)$, the lower bound is only $H(M_{12})\geq1.5$. We also show that this protocol is randomness-optimal, which proves a recent conjecture of \cite{AbbeLe14} for three users.

For vectors $U,V\in\{0,1,2\}^n$, we write $U+V$ to denote the component-wise addition modulo-3.
\begin{figure}[htb]
\hrule height 1pt
\vspace{.05cm}
{\bf Algorithm 4:} {Secure Computation of \textsc{sum}}
\hrule
\begin{algorithmic}[1]
\REQUIRE Alice and Bob have input vectors $X,Y\in\{0,1\}^n$.
\ENSURE Charlie securely computes the component-wise {\sc sum} $Z=X+Y$.

\medskip

\STATE Charlie samples $n$ i.i.d. uniformly distributed elements $K=(K_1,K_2,\hdots,K_n)$ from $\{0,1,2\}$ using his private randomness; sends it to Alice as $M_{\vec{31}}=K$.

\STATE Alice sends $M_{\vec{12}}=M_{\vec{31}}+X$ to Bob.

\STATE Bob sends $M_{\vec{23}}=M_{\vec{12}}+Y$ to Charlie.

\STATE Charlie outputs $Z=M_{\vec{23}}-K$.
\end{algorithmic}
\hrule
\caption{A protocol to compute {\sc sum}. The protocol requires a ternary symbol to be exchanged over all the three links per computation. We show a lower bound of $\log(3)$ both on Alice-Charlie and Bob-Charlie links and a lower bound of 1.5 on Alice-Bob link.}
\label{fig:sum}
\end{figure}
\begin{thm}\label{thm:sum}
Any secure protocol for computing {\sc sum}, where $p_{XY}$ has full support over $\{0,1\}^n\times\{0,1\}^n$ must satisfy 
\begin{align*}
H(M_{31}),H(M_{23}) &\geq n\log(3) \text{ and } H(M_{12})\geq 1.5n, \\
\rho(\text{\sc sum}) &\geq n\log(3).
\end{align*}
\end{thm}
\begin{proof}
It is easy to see that \textsc{sum} satisfies Condition~1 and Condition~2 of \Lemmaref{XYZ_inde_M123}. Also, $RI(Y;Z)=I(Y;Z)$ and $RI(Z;X)=I(Z;X)$. It turns out that for $H(M_{31})$ and $H(M_{23})$, the bounds in~\eqref{eq:improved_prelim_lb_M31} and~\eqref{eq:improved_prelim_lb_M23} are better than~\eqref{eq:M31_WithCond} and~\eqref{eq:M23_WithCond} respectively. Since $X$ can be determined from $(Y,Z)$ and $Y$ can be determined from $(X,Z)$, we can simplify the bounds in~\eqref{eq:improved_prelim_lb_M31} and~\eqref{eq:improved_prelim_lb_M23} to the following:
\begin{align*}
H(M_{31}) &\geq \sup_{p_{X'Y'}} \left( H(Z') \right), \\
H(M_{23}) &\geq \sup_{p_{X'Y'}} \left( H(Z') \right).
\end{align*}
For $H(M_{31})$, taking $p_{X'Y'}(0,0)$ = $p_{X'Y'}(1,1)=1/3$ and $p_{X'Y'}(0,1)$ = $p_{X'Y'}(1,0)=1/6$ gives $H(M_{31}),H(M_{23})\geq n\log(3)$. For $H(M_{12})$, the bound in \eqref{eq:M12_bound} is better than \eqref{eq:improved_prelim_lb_M12} and \eqref{eq:M12_bound} simplifies to
\[ H(M_{12}) \geq \sup_{p_{X'}}\left\{\sup_{p_{Y'}} I(Y';Z') + \sup_{p_{Y''}}\left\{I(X';Z'') + H(X',Y''|Z'')\right\}\right\}.\]
The second term simplifies to $H(X')$. Taking $X',Y'\sim$ Bern(1/2) gives $H(M_{12})\geq 1.5n$.

Since {\sc sum} satisfies condition 1 of \Lemmaref{XYZ_inde_M123}. So, from~\Theoremref{randomness}, we have $\rho(\text{\sc sum}) \geq H(M_{31})$, which from the above calculation is lower bounded by $n\log(3)$, implying the randomness-optimality of the above protocol.

\end{proof}

\subsection{Secure Computation of {\sc controlled erasure}}
\label{app:erasure}

The controlled erasure function from~\cite{DPAllerton13} is shown below. Alice's input $X$ acts as
the ``control'' which decides whether Charlie receives an erasure
($\Delta$) or Bob's input $Y$. 
\begin{center}
\begin{tabular}{ccc}
\toprule
& \multicolumn{2}{c}{y}\\
\cmidrule(r){2-3} 
x&\multicolumn{1}{c}{0}&\multicolumn{1}{c}{1}\\
\midrule
0& $\Delta$ & $\Delta$\\
1& 0 & 1\\
\bottomrule
\end{tabular}
\end{center}
Notice that Charlie always find out Alice's control bit, but does not learn
Bob's bit when it is erased. This function does not satisfy Condition~1 of
\Lemmaref{XYZ_inde_M123}.

\Figureref{erasure} gives a protocol (repeated
from~\cite{DPAllerton13}) for securely computing this
function on each location of strings of length $n$. Bob sends his input
string to Charlie under the cover of a one-time pad and reveals the key
used to Alice. Alice sends his input to Charlie compressed using a Huffman
code (replaced by Lempel-Ziv if we want the protocol to be distribution
independent). He also sends to Charlie those key bits he received
corresponding to the locations where there is no erasure (i.e., where his
input bit is 1). When $X\sim\text{Bernoulli}(p)$ and
$Y\sim\text{Bernoulli}(q)$, i.i.d., where $p,q\in(0,1)$, the expected
message length for Alice-Charlie link is ${\mathbb E}[L_{31}] < nH_2(p)+1
+ np$, the messages lengths on the other two links are determinisitically
$n$ each, $L_{12}=L_{23}=n$. Here we prove the optimality of this protocol
for $X\sim\text{Bernoulli}( p)$ and $Y\sim\text{Bernoulli}( q)$, where $p,q
\in (0,1)$; \cite{DPAllerton13}~only considered the case where $X,Y\sim
\text{Bernoulli}(1/2)$. We also prove that this protocol is randomness-optimal.
\begin{figure}[htb]
\hrule height 1pt
\vspace{.05cm}
{\bf Algorithm 2:} {Secure Computation of \textsc{controlled erasure}} 
\hrule
\begin{algorithmic}[1]
\REQUIRE Alice \& Bob have input bits $X^n,Y^n\in\{0,1\}^n$.
\ENSURE Charlie securely computes the {\sc controlled erasure} function
\[Z_i=f(X_i,Y_i),\qquad i=1,\ldots,n.\]

\medskip

\STATE Bob samples $n$ i.i.d. uniformly distributed bits $K^n$ from his private randomness; sends it to Alice as $M_{\vec{21},1}=K^n$. Bob sends to Charlie his input $Y^n$ masked (bit-wise) with $K^n$ as $M_{\vec{23},1}=Y^n\oplus K^n$. 

\STATE Alice sends his input $X^n$ to Charlie compressed using a Huffman
code (or Lempel-Ziv if we want the protocol to not depend on the input
distribution of $X^n$); let $c(X^n)$ be the codeword. Alice also sends to
Charlie the sequence of key bits $K_i$ corresponding to the locations where
his input $X_i$ is 1.
\[ M_{\vec{12},2}= c(X^n),(K_i)_{i:X_i=1}.\]

\STATE Charlie outputs \[Z_i=\left\{ \begin{array}{ll}\Delta,& \text{ if } X_i=0\\ (Y_i\oplus K_i)\oplus K_i,& \text{ if } X_i=1.\end{array}\right.\]
\end{algorithmic}
\hrule

\caption{A protocol to compute {\sc controlled erasure} function. For
$X\sim\text{Bernoulli}(p)$ and $Y\sim\text{Bernoulli}(q)$, both i.i.d and
$p,q\in(0,1)$, the
expected message lengths per bit are ${\mathbb E}[L_{31}] < n(H_2(p)+p)+1$,
$L_{12} = n$, and $L_{23} =n$. We show that these are asymptotically
optimal by showing the following lower bounds: $H(M_{31}) \geq n(H_2(p) +
p)$, $H(M_{12})\geq n$ and $H(M_{23})\geq n$.}
\label{fig:erasure}
\end{figure}

\begin{thm}\label{thm:erasure}
Any secure protocol for computing {\sc controlled erasure} for $X\sim \text{Bernoulli}(p )$ and $Y\sim \text{Bernoulli}(q )$, both i.i.d., with $p,q\in(0,1)$ over block length $n$ must satisfy
\begin{align*}
H(M_{31}) \geq n(H_2(p) + p),\quad H(M_{12}) &\geq n,\;\; \text{ and }\quad H(M_{23})\geq n, \\
\rho(\text{\sc controlled-erasure}) &\geq n.
\end{align*}
\end{thm}
\begin{proof}
It is easy to see that this function satisfies only Condition~2 of \Lemmaref{XYZ_inde_M123}. We also have $RI(X;Z)=0$ and $RI(Y;Z)=I(Y;Z)$ for this function. Since Condition~1 of \Lemmaref{XYZ_inde_M123} is not satisfied, our best bound for $H(M_{31})$ is given by \eqref{eq:M31_NoCond}. Since $X$ is independent of $Y''$ in \eqref{eq:M31_NoCond} and we are computing a deterministic function, the bound in \eqref{eq:M31_NoCond} simplifies to the following:
\begin{align*}
H(M_{31}) &\geq \sup_{p_{Y'}}\left\{I({Y'}^n;{Z'}^n)\right\} + H(X^n).
\end{align*}
The optimum bound for $H(M_{31})$ is obtained by taking $Y'\sim \text{Bernoulli}(1/2)$, which gives $H(M_{31}) \geq n(p+H_2(p ))$. For $H(M_{23})$, we can apply the bound in \eqref{eq:improved_prelim_lb_M23}, which simplifies to the following:
\begin{align*}
H(M_{23}) &\geq \sup_{p_{X'Y'}}\left\{H({Y'}^n | {X'}^n)\right\}.
\end{align*}
Taking $Y'$ to be independent of $X'$ and $Y'\sim \text{Bernoulli}(1/2)$ gives $H(M_{31}) \geq n(p+H_2(p ))$. For $H(M_{12})$, we can apply the bound in \eqref{eq:improved_prelim_lb_M12}, which simplifies to the following:
\begin{align*}
H(M_{12}) &\geq \sup_{p_{X'Y'}} \left(I(Y';Z') + H(X',Y'|Z')\right) \\
&= \sup_{p_{X'Y'}} \left(-H(Z' | Y') + H(X',Y')\right) \\
&\stackrel{\text{(a)}}{=} \sup_{p_{X'Y'}} \left(-H(X' | Y') + H(X',Y')\right) \\
&= \sup_{p_{X'Y'}} \left(H(X')\right),
\end{align*}
where (a) follows because, we can determine $X$ from $Z$ and $Z$ is a deterministic function of $(X,Y)$. Now, taking $X'\sim \text{Bernoulli}(1/2)$ gives $H(M_{12})\geq n$.

From~\Theoremref{randomness} and the above bound on $H(M_{12})$, we have $\rho(\text{\sc controlled-erasure}) \geq n$, which implies that the above protocol is randomness-optimal.

\end{proof}

\subsection{Secure Computation of {\sc remote $\binom{m}{1}$-OT$^n_2$}}
\label{app:ot}

\begin{figure}[htb]
\hrule height 1pt
\vspace{.05cm}
{\bf Algorithm 3:} {Secure Computation of \textsc{remote $\binom{m}{1}$-OT$^n_2$}} 
\hrule
\begin{algorithmic}[1]
\REQUIRE Alice has $m$ input bit strings $X_0,X_1,\hdots,X_{n-1}$ each of length $n$ \& Bob has an input $Y \in \{0,1,\hdots,m-1\}$.
\ENSURE Charlie securely computes the {\sc remote $\binom{m}{1}$-OT$^n_2$}: $Z=X_Y$.

\medskip

\STATE Alice samples $nm + \log m$ indep., uniformly distributed bits from her private randomness. Denote the first $m$ blocks each of length $n$ of this random string by $K_0,K_1,\hdots,K_{m-1}$ and the last $\log m$ bits by $\pi$. Alice sends it to Bob as $M_{\vec{12},1}=(K_0,K_1,\hdots,K_{m-1},\pi)$.

\STATE Alice computes $M^{(i)}=X_{\pi+i\; (\text{mod }m)}\oplus K_{\pi+i\; (\text{mod } m)}, \quad i\in\{0,1,\hdots,m-1\}$ and sends to Charlie 
$M_{\vec{13},2}=(M^{(0)},M^{(1)},\hdots,M^{(m-1)})$. Bob computes $C=Y-\pi\; (\text{mod } m), K=K_Y$ and sends to Charlie $M_{\vec{23},2}=(C, K)$.

\STATE Charlie outputs $Z=M^{(C )}\oplus K$. 

\end{algorithmic}

\hrule
\caption{A protocol to securely compute {\sc remote $\binom{m}{1}$-OT$^n_2$}, which is a special case of the general protocol given in \cite{FeigeKiNa94}. The protocol requires $nm$ bits to be exchanged over the Alice-Charlie (31) link, $n+\log m$  bits over the Bob-Charlie (23) link and $nm+\log m$ bits over the Alice-Bob (12) link. We show optimality of our protocol by showing that any protocol must exchange an expected $nm$ bits  over the Alice-Charlie (31) link, $n+\log m$ bite over the Bob-Charlie (23) link and $nm+\log m$ bits over the Alice-Bob (12) link.}
\label{fig:ot}
\end{figure}
\begin{proof}[Proof of \Theoremref{ot}]
\textsc{remote $\binom{m}{1}$-OT$^n_2$} satisfies Condition~1 and Condition~2 of \Lemmaref{XYZ_inde_M123}. We also have, $RI(Y;Z)=I(Y;Z)$ and $RI(Z;X)=I(Z;X)$. It turns out that for $H(M_{31})$ and $H(M_{23})$,~\eqref{eq:improved_prelim_lb_M31} and~\eqref{eq:improved_prelim_lb_M23} give the same bounds as~\eqref{eq:M31_WithCond} and~\eqref{eq:M23_WithCond} respectively. We will consider the bounds in~\eqref{eq:M31_WithCond} and~\eqref{eq:M23_WithCond} in the following. Since $X'$ is independent of $Y''$ in \eqref{eq:M31_WithCond}, $X''$ is independent of $Y'$ in \eqref{eq:M23_WithCond} and we are computing a deterministic function, the bounds in \eqref{eq:M31_WithCond} and \eqref{eq:M23_WithCond} simplify to the following:

\begin{align*}
H(M_{31}) \geq \sup_{p_{X'}} \left\{ \left( \sup_{p_{Y'}} I(Y';Z') \right) + H(X') \right\}, \\
H(M_{23}) \geq \sup_{p_{Y'}} \left\{ \left( \sup_{p_{X'}} I(X';Z') \right) + H(Y') \right\}.
\end{align*}
Taking $X'$ and $Y'$ to be uniform, we get $H(M_{31}) \geq nm$.
To derive a lower bound on $H(M_{23})$, take $Y' \sim \text{unif}\{0,1\}$ and $X'$ distributed as below
\[ p_{X_0',X_1',\hdots,X_{m-1}'}(x_0,x_1,\hdots,x_{m-1})=\left\{\begin{array}{ll} \frac{1}{2^n}-\epsilon, &
x_0=x_1=\hdots=x_{m-1}\\ \epsilon/(2^{n(m-1)}-1),& \text{otherwise},\end{array}\right.\]
where $\epsilon>0$ can be made arbitrarily small to make $I(Z';X')$ as close
to $n$ as desired. This gives a bound of $H(M_{23})\geq n+\log m$. For $H(M_{12})$, the bottom row of \eqref{eq:M12_bound} simplifies to
\[ H(M_{12}) \geq \sup_{p_{Y'}}\left\{\sup_{p_{X'}} I(X';Z') + \sup_{p_{X''}}\left\{ I(Y';Z'') + H(X'',Y'|Z'')\right\}\right\}.\]
Taking $Y'$ and $X''$ to be uniform and $X'$ to be as below
\[ p_{X_0',X_1',\hdots,X_{m-1}'}(x_0,x_1,\hdots,x_{m-1})=\left\{\begin{array}{ll} \frac{1}{2^n}-\epsilon, &
x_0=x_1=\hdots=x_{m-1}\\ \epsilon/(2^{n(m-1)}-1),& \text{otherwise},\end{array}\right.\]
where $\epsilon>0$ can be made arbitrarily small to make $I(X';Z')$ as close to $n$ as desired. This gives a bound of $H(M_{12})\geq nm+\log m$.

From~\Theoremref{randomness} and the above bound on $H(M_{12})$, we have $\rho(\text{\sc remote-ot}) \geq nm+\log m$, which implies that the above protocol is randomness-optimal.

\end{proof}


\subsection{Secure Computation of {\sc and}}\label{app:and}

\begin{figure}[htb]
\hrule height 1pt
\vspace{.05cm}
{\bf Algorithm 5:} {Secure Computation of \textsc{and}}
\hrule
\begin{algorithmic}[1]
\REQUIRE Alice has an input bit $X$ \& Bob has a bit $Y$.
\ENSURE Charlie securely computes the {\sc and} $Z=X\wedge Y$.

\medskip

\STATE Alice samples a uniform random permutation $(\alpha,\beta,\gamma)$
of $(0,1,2)$ from her private randomness; sends it to Bob
$M_{\vec{12}}=(\alpha,\beta,\gamma)$ (using a symbol from an alphabet of
size 6).

\STATE Alice sends $\alpha$ to Charlie if $X=1$, and $\beta$ if $X=0$. Bob
sends $\alpha$ to Charlie if $Y=1$, and $\gamma$ if $Y=1$.
\begin{align*}
M_{31}=\begin{cases}
\alpha \qquad \text{if }X=1 \\
\beta  \qquad \text{if }X=0
\end{cases}
&
\qquad
M_{23}=\begin{cases}
\alpha \qquad \text{if }Y=1 \\
\gamma  \qquad \text{if }Y=0
\end{cases}
\end{align*}

\STATE Charlie outputs $Z=1$ if $M_{31}=M_{23}$, and 0 otherwise.
\end{algorithmic}
\hrule
\caption{A protocol to compute {\sc and} \cite{FeigeKiNa94}. The protocol requires a ternary symbol to be exchanged over the Alice-Charlie (31) and Bob-Charlie (23) links and symbols from an alphabet of size 6 over the Alice-Bob (12) link per {\sc and} computation.}
\label{fig:and}
\end{figure}
\begin{proof}[Proof of \Theoremref{and}]
We will prove the result only for $n=1$, i.e., when input consists of only one bit. The result for general $n$ follows by taking in the following proof $X_{i}'$s, $Y_{i}'$s and $Y_{i}''$s to be i.i.d.

It is easy to see that \textsc{and} satisfies Condition~1 and Condition~2 of \Lemmaref{XYZ_inde_M123}. Also, $RI(Y;Z)=I(Y;Z)$ and $RI(Z;X)=I(Z;X)$. It turns out that for $H(M_{31})$ and $H(M_{23})$, the bounds in~\eqref{eq:improved_prelim_lb_M31} and~\eqref{eq:improved_prelim_lb_M23} are better than~\eqref{eq:M31_WithCond} and~\eqref{eq:M23_WithCond} respectively. The simplified bounds in~\eqref{eq:improved_prelim_lb_M31} and~\eqref{eq:improved_prelim_lb_M23} are as follows:
\begin{align*}
H(M_{31}) \geq \sup_{p_{X'Y'}} \left( I(Y';Z') + H(X',Z'|Y') \right), \\
H(M_{23}) \geq \sup_{p_{X'Y'}} \left( I(X';Z') + H(Y',Z'|X') \right).
\end{align*}
For $H(M_{31})$, take $p_{X'Y'}(0,0)=p_{X'Y'}(1,0)=p_{X'Y'}(1,1)=(1-\epsilon)/3$ and $p_{X'Y'}(0,1)=\epsilon$, where $\epsilon>0$ can be made arbitrarily small to make $H(M_{31})$ as close to $\log(3)$ as we desire.

For $H(M_{23})$, take $p_{X'Y'}(0,0)=p_{X'Y'}(0,1)=p_{X'Y'}(1,1)=(1-\epsilon)/3$ and $p_{X'Y'}(1,0)=\epsilon$, where $\epsilon>0$ can be made arbitrarily small to make $H(M_{23})$ as close to $\log(3)$ as we desire.

For $H(M_{12})$, \eqref{eq:M12_bound} simplifies to
\[ H(M_{12}) \geq \sup_{p_{X'}}\left\{\sup_{p_{Y'}} I(Y';Z') + \sup_{p_{Y''}}\left\{I(X';Z'') + H(X',Y''|Z'')\right\}\right\}.\]
The second term simplifies to $H(X') + p_{X'}(0)$ by taking $Y''$ to be uniform. 
Taking $p_{X'}(1)=0.456$ and $p_{Y'}(1)=0.397$ gives $H(M_{12}) \geq
1.826$.

From~\Theoremref{randomness} and the above bound on $H(M_{12})$, we have $\rho(\text{\sc AND}) \geq n(1.826)$, whereas the protocol requires $1+\log 3$ random bits.

\end{proof}

{\bf Note:} We need the use of \Lemmaref{infoineq} (information inequality) only to improve the bound on $H(M_{12})$ in {\sc remote-ot, sum} and {\sc and}. All other bounds in all other functions do not need the use information inequality.

\section{Lower Bounds for Dependent Inputs}\label{app:dependent_inputs_lbs}
We will show that all our lower bounds proven for independent inputs hold
for dependent inputs as well provided the distribution has full support. In \Subsectionref{improved_lbs}, we observed that any secure protocol $\Pi(p_{XY},p_{Z|XY})$, where distribution $p_{XY}$ has full support, continues to be a secure protocol even if we switch the input distribution to a different one $p_{\widetilde{X}\widetilde{Y}}$.

Since we can switch to any distribution $p_{\widetilde{X}\widetilde{Y}}$, in particular, we can switch to $p_{\widetilde{X}\widetilde{Y}}$, where $\widetilde{X}$ and $\widetilde{Y}$ have the same marginals as $X$ and $Y$ respectively, i.e., $p_{\widetilde{X}}(x)=p_X(x), \forall x\in\mathcal{X}$ and $p_{\widetilde{Y}}(y)=p_Y(y), \forall y\in\mathcal{Y}$. This allows us to argue that the communication lower bounds for $\Pi(p_{\widetilde{X}\widetilde{Y}},p_{Z|XY})$ also hold for $\Pi(p_{XY},p_{Z|XY})$. To prove this, we show below that the resulting marginal distributions on the transcripts remain the same as the original ones, implying the same entropies.

Let denote the resulting distribution on the transcript on 12 link by $\widetilde{M}_{12}$ and similarly on the other two links.
\begin{align*}
p_{\widetilde{M}_{12}}(m_{12}) &= \sum_{x,y}p_{\widetilde{M}_{12}|\widetilde{X}\widetilde{Y}}(m_{12}|x,y)p_{\widetilde{X}\widetilde{Y}}(x,y) \\
&\stackrel{\text{(a)}}{=} \sum_{x,y}p_{M_{12}|XY}(m_{12}|x,y)p_{\widetilde{X}\widetilde{Y}}(x,y) \\
&\stackrel{\text{(b)}}{=} \sum_{x,y}p_{M_{12}|X}(m_{12}|x)p_{\widetilde{X}\widetilde{Y}}(x,y) \\
&= \sum_{x}p_{M_{12}|X}(m_{12}|x)\sum_{y}p_{\widetilde{X}\widetilde{Y}}(x,y) \\
&= \sum_{x}p_{M_{12}|X}(m_{12}|x) p_{\widetilde{X}}(x) \\
&= \sum_{x}p_{M_{12}|X}(m_{12}|x) p_{X}(x) \\
&= p_{M_{12}}(m_{12}),
\end{align*}
where (a) follows from the fact that in a secure computation protocol, once
Alice and Bob are given inputs $X=x$ and $Y=y$ respectively, the protocol
produces $(m_{12},m_{23},m_{31},z)$ according to the conditional
distribution $p_{M_{12}M_{23}M_{31}Z|XY}(m_{12},m_{23},m_{31},z|x,y)$ and
this conditional distribution does not depend on the distribution $p_{XY}$,
hence, $p_{\widetilde{M}_{12}|\widetilde{X}\widetilde{Y}}(m_{12}|x,y) = p_{M_{12}|XY}(m_{12}|x,y)$; and (b)
follows from privacy against Alice. This implies that $H(\widetilde{M}_{12})=H(M_{12})$.
Similarly we can prove $H(\widetilde{M}_{23})=H(M_{23})$ and $H(\widetilde{M}_{31})=H(M_{31})$.

Proofs of our lower bounds for secure computation in \Theoremref{lowerbound} and \Theoremref{conditionallowerbounds} assumed independent inputs. For them to hold for dependent ones, we can take $p_{\widetilde{X}\widetilde{Y}}$ in above to be a product distribution $p_{\widetilde{X}\widetilde{Y}}=p_{\widetilde{X}}p_{\widetilde{Y}}$ with $\widetilde{X}$ and $\widetilde{Y}$ having the same marginals as $X$ and $Y$.

\section{Dependence on Input Distributions} \label{app:discuss}
Our communication
lower bounds were developed for protocols whose designs may take into
account the distributions of $X$ and $Y$. Specifically, the right hand sides of \eqref{eq:improved_prelim_lb_M12} and \eqref{eq:M12_bound} do not depend on the distributions $p_Xp_Y$ of the inputs. Thus, even though we
allow the protocol to depend on the distributions, our lower bound on
$H(M_{12})$ does not. The same is true for \eqref{eq:improved_prelim_lb_M31} and \eqref{eq:M31_WithCond} for $H(M_{31})$ and \eqref{eq:improved_prelim_lb_M23} and \eqref{eq:M23_WithCond} for $H(M_{23})$, which apply when the function satsifies certain conditions. As the {\sc controlled-erasure} example 
(\Appendixref{erasure}) demonstrates, when these conditions are
not satisfied, the communication complexity of the optimal protocol may
indeed depend on the distribution of the input.
Notice that the specific protocols we have given in this paper do not need
the knowledge of the input distributions.

\end{document}